\documentclass[10pt]{article}
\usepackage{blindtext} % needed for creating dummy text passages
\usepackage{amsmath} % needed for command eqref
\usepackage{amssymb} % needed for math fonts
\usepackage[colorlinks=true,breaklinks]{hyperref} % needed for creating hyperlinks in the document, the option colorlinks=true gets rid of the awful boxes, breaklinks breaks lonkg links (list of figures), and ngerman sets everything for german as default hyperlinks language
\usepackage{xcolor}
\definecolor{c1}{rgb}{0,0,1} % blue
\definecolor{c2}{rgb}{0,0.3,0.9} % light blue
\definecolor{c3}{rgb}{0.3,0,0.9} % red blue
\hypersetup{
    linkcolor={c1}, % internal links
    citecolor={c2}, % citations
    urlcolor={c3} % external links/urls
}
\usepackage{fullpage}
\usepackage{graphicx}
\usepackage{amsthm} 
\usepackage{dsfont}
\usepackage{array}
\usepackage{multicol}
\usepackage[utf8]{inputenc}  
\usepackage[T1]{fontenc}  
\usepackage[frenchb,english]{babel}
\usepackage{MnSymbol}
\usepackage[hyphenbreaks]{breakurl}
\usepackage{multicol}
\usepackage{pdfpages}
\usepackage{tikz}
\usepackage{tikz-cd}
\usepackage[blocks]{authblk}
\usepackage{braket}
\usepackage{cite}
\usepackage{epigraph}	
\usepackage{graphicx} % needed for \includegraphics 
\usepackage{enumerate} % needed for some options in enumerate
\usepackage{todonotes} % needed for todos

\usepackage[top=2cm, bottom=1.8cm,left=2.5cm,right=2.5cm,head=2cm,headsep=0.5cm]{geometry} % needed for page border settings
\usepackage{fancyhdr} % needed for head and foot options
\usetikzlibrary{arrows}
\usetikzlibrary{decorations.pathmorphing}

\theoremstyle{plain}

\newtheorem{de}{Definition}[section]

\newtheorem{de*}{Definition}%[section]
\newtheorem{theo}{Theorem}[section] 
 
\newtheorem{prop}{Proposition}[section]                              
\newtheorem{lem}{Lemma}[section]

\newtheorem{coro}{Corollary}[section]

\theoremstyle{definition}
\newtheorem{ex}{Examples}[section]
\newtheorem{example}{Examples}[section]

\theoremstyle{remark}
\newtheorem{rem}{Remarks}[section]

\newcommand{\Hcal}{\mathcal{H}}
\newcommand{\Kcal}{\mathcal{K}}
\newcommand{\Ncal}{\mathcal{N}}
\newcommand{\DF}{\mathcal{N}(\mathcal{P})}

\newcommand{\Lcal}{\mathcal{L}}

\newcommand{\Pcal}{\mathcal{P}}
\newcommand{\Acal}{\mathcal{A}}
\newcommand{\Ccal}{\mathcal{C}}
\newcommand{\Scal}{\mathcal{S}}

\newcommand{\Ecal}{\mathcal{E}}
\newcommand{\Bcal}{\mathcal{B}}

\newcommand{\R}{\mathbb{R}}

\newcommand{\Nbb}{\mathbb{N}}
\newcommand{\Lbb}{\mathbb{L}}

\newcommand{\C}{\mathbb{C}}

\newcommand{\Sp}{\text{Sp}}
\newcommand{\deco}{\text{deco}}

\newcommand{\scal}[2]{\left\langle #1 \,\big|\,#2\right\rangle}
\newcommand{\sca}[2]{\left\langle #1\, ,\,#2\right\rangle}
\newcommand{\norm}[1]{\left\| #1 \right\|}
\newcommand{\proj}[2]{|#1\rangle\langle #2|}
\newcommand{\Dent}[2]{\operatorname{D}\left(#1\,||\,#2\right)}
\newcommand{\DentDF}[1]{\operatorname{D}\left(#1\,,\,\Ncal\right)}
\newcommand{\IM}[1]{\operatorname{I}\left(\operatorname{A}:\operatorname{B}\right)_{#1}}

\newcommand{\eps}{\varepsilon}

\newcommand{\Ind}{\mathds{1}}

\newcommand{\Tr}{\operatorname{Tr}\,}
\newcommand{\tr}{\operatorname{Tr}}
\newcommand{\Var}{\operatorname{Var}}

\newcommand{\EP}{\operatorname{EP}}
\newcommand{\IP}{\operatorname{IP}}

\newcommand{\Ker}{\operatorname{Ker}\,}

\numberwithin{equation}{section}

\begin{document}

\title{Estimating the decoherence time using non-commutative Functional Inequalities\thanks{Work supported by A.N.R. grant: ANR-14-CE25-0003 "StoQ" }}
%\titlerunning{Decoherence time and non-commutative Functional Inequalities}

\author{Ivan Bardet}
    \affil{\small Institut des Hautes \'Etudes Scientifiques, Université Paris-Saclay, 35 Route de Chartres, 91440 Bures-sur-Yvette, France}
    
\date{October 3, 2017}
%\keywords{Non-commutative functional inequalities, environment-induced decoherence}

%\authorrunning{I. B.}

\maketitle

\begin{abstract} 
We generalize the notions of the non-commutative Poincar\'e and modified log-Sobolev inequalities for primitive quantum Markov semigroups (QMS) to not necessarily primitive ones. These two inequalities provide estimates on the decoherence time of the evolution. More precisely, we focus on an algebraic definition of environment-induced decoherence in open quantum systems which happens to be generic on finite dimensional systems and describes the asymptotic behavior of any QMS. An essential tool in our analysis is the explicit structure of the decoherence-free algebra generated by the QMS, a central object in the study of passive quantum error correction schemes. The Poincaré constant corresponds to the spectral gap of the QMS, which implies its positivity, while we prove that the modified log-Sobolev constant is positive under the $\Lbb_1$-regularity of the Dirichlet form, a condition that also appears in the primitive case. We furthermore prove that strong $\Lbb_p$-regularity holds for quantum Markov semigroups that satisfy a strong form of detailed balance condition for $p\geq1$. The latter condition includes all known cases where this strong regularity was proved. Finally and to emphasize the mathematical interest of this study compared to the classical case, we focus on two truly quantum scenarios, one exhibiting quantum coherence, and the other, quantum correlations.
\end{abstract}

\section{Introduction}
Functional Inequalities (FI) such as the Poincaré Inequality (PI) and the log-Sobolev Inequality (LSI) play a central role in the study of the asymptotic behavior of open systems. They were introduced for quantum systems in the pioneering article \cite{O-Z} and since then have been the subject of intensive studies \cite{C-M,C-S,CKMT15,MSW2016,MSW16,TPK,KT2013}. In all those works, the dissipative evolution of an open quantum system is assumed to drive the system toward its unique invariant state: this is the \emph{primitive} assumption. The goal of this article is to initiate the study of functional inequalities for not necessarily primitive quantum Markov semigroups, starting with the PI and the modified LSI. There are several motivations behind such a generalization.

\paragraph{}Functional inequalities are particularly relevant in order to prove the rapid mixing of Markovian dynamics \cite{KT2013}. The role of this latter property is well-known in statistical physics where it is used in order to prove the existence of a unique invariant state for local dissipative systems \cite{Zeg90log,Zeg90log2,Zeg92,MO94,Mart99}. It was also recently highlighted as a key assumption in order to prove stability results for such systems \cite{CLMP-G13,LCMP-G15}, which is a much desired property of protocols involving dissipative engineering \cite{KBDKMZ08,VWC09,KRS11}. In this context, it was already mentioned in \cite{CLMP-G13} that withdrawing the primitive assumption would lead to further technical subtleties. This article is meant to establish the bases for such an extension.

\paragraph{}A natural generalization of the primitive assumption is given by the notion of \emph{Environment-Induced Decoherence} (EID). Introduced by Zurek in the eighties \cite{Zur81,Zur82}, EID provides a dynamical argument to the disappearance of quantum states in our classical world. The first mathematical formulation of EID was proposed by Blanchard and Olkiewicz in \cite{B-O} (see also \cite{H11,CSU2,CSU4} for discussions and more recent formulations). In this article, we focus on quantum systems with finite degrees of freedom. Under the legitimate Markovian assumption, one can consider that the evolution of the system is modeled by a quantum Markov semigroup (QMS) $(\Pcal_t)_{t\geq0}$ acting on the algebra of (bounded) operators $\Bcal(\Hcal)$, for some finite dimensional Hilbert space $\Hcal$. That is, $\Pcal$ is a continuous semigroup of completely positive maps, that preserve the identity operator: $\Pcal_t(I_\Hcal)=I_\Hcal$ for all $t\geq0$. For such systems, it was proved that the algebra of observables can be split into two parts \cite{H11,CSU2,DFSU1}:
\begin{itemize}
\item a (von Neumann) subalgebra of $\Bcal(\Hcal)$, called the \emph{Decoherence-Free Algebra} (DF algebra), which is free from any dissipative effect: the evolution of any observable belonging to this subalgebra is given by a $*$-isomorphism, which in our case simply means a unitary evolution;
\item a subspace which is not detectable by any experiment: the observation of any observable belonging to this subspace gives a zero-value in expectation, in the long-time asymptotic, whatever the state of the system is.
\end{itemize}
The second point will be primordial in our analysis, as it provides an alternative to the primitive case where any state convergence to the unique invariant state of the evolution. Thus, in the general case, the dissipative effects induced by the environment drive the observables of the system toward a closed subsystem modeled by the DF algebra. When this algebra is not trivial, there is necessarily an infinity of invariant states and in the long-time behavior, the state of the system can continue to evolve according to a unitary evolution. When the QMS is primitive, the DF algebra is trivial and reduced to the multiples of the identity operator. In this case, all states converge toward the unique invariant state. 

\paragraph{}Apart from providing an alternative to the primitive assumption, EID also plays a determinant role in quantum information theory and quantum control theory. The DF algebra is a particular instance of decoherence-free subsystems, one of the most promising and physically relevant mathematical tools proposed in order to overcome decoherence in quantum computing \cite{LCW98,KBLW01}. Decoherence-free subspaces and subsystems could play a decisive role in the achievement of quantum information processing \cite{AFR14,CT15,LW03}, as they identify regions of the system that are at the same time protected from decoherence and still big enough to allow universal quantum computations \cite{KBLW01}. Therefore, if rapid mixing is responsible for the stability of state preparation protocols based on dissipative engineering (as in \cite{CLMP-G13,LCMP-G15}), than one can hope that a generalization of this concept to non-primitive QMS would similarly permit to proved the stability of "decoherence-free" protocols for quantum computing (as in \cite{LCW98,KBLW01}). In this article, we initiate the investigation of new functional inequalities that imply "rapid decoherence", in the sense that decoherence occurs exponentially fast.

\paragraph{Motivation on one simple example:}
Following the lines of \cite{KT2013} and in order to motivate our work, let focus on one simple situation that perfectly captures the idea of decoherence. Consider the following evolution on the Hilbert space $\Hcal=\C^d$ in the Schrödinger picture:
\begin{equation}\label{eq_def_QMSdeco}
\rho_t=e^{-\gamma t}\,\rho+(1-e^{-\gamma t})\,E_{\Acal}(\rho)\,,
\end{equation}
where $\gamma$ is a positive constant (modeling the strength of the interaction with the environment or the measuring device) and $E_{\Acal}$ is the orthogonal projection on the commutative subalgebra $\Acal\simeq\C^d$ (the algebra of diagonal operators in a certain basis of $\C^d$), for the Hilbert-Schmidt scalar product. This defines a proper QMS $\Pcal^\deco_{*t}(\rho)=\rho_t$\footnote{Here the $*$ refers to the predual of the QMS $\Pcal=(\Pcal_t)_{t\geq0}$. The precise definition will be given in the next section.} that we call the \emph{Decoherence Quantum Markov Semigroup}, as it perfectly reflects the idea that the dissipation induced by the environment will cause the state of the system to collapse into a "classical" state: one has clearly:
\begin{equation}\label{eq_limitQMSdeco}
\underset{t\to+\infty}{\lim}\,\Pcal^\deco_{*t}\left(\rho-E_{\Acal}(\rho)\right)=0\,.
\end{equation}

Further remark that the maximally mixed density matrix $\frac{I_\Hcal}{d}$ is an invariant density matrix and that $\Pcal^\deco$ is a symmetric operator with respect to the Hilbert-Schmidt scalar product. One usually says that $\Pcal^\deco$ satisfies the \emph{Detailed Balance Condition} with respect to $\frac{I_\Hcal}{d}$.\\
Equation \eqref{eq_limitQMSdeco} states that in the long time behavior, quantum correlations disappear as the off-diagonal terms vanish. We can then define the \emph{decoherence time} as
\begin{equation}\label{eq_def_decotime}
\tau_\deco(\eps)=\min\,\left\{t\geq0\,;\,\norm{\Pcal^\deco_{*t}\left(\rho-E_{\Acal}(\rho)\right)}_{\Tr}\leq \eps\quad\forall \rho\right\}\,,
\end{equation}
where $\norm{\cdot}_{\tr}$ is the trace norm, $\norm{X}_{\tr}=\Tr\sqrt{X^*X}$, and where $\eps$ is a positive constant. The use of the trace norm is justified by its operational interpretation as a measure of distinguishability between two states. As proved for instance in \cite{CSU3}, the limit in Equation \eqref{eq_limitQMSdeco} is generic, in the sense that it holds for any QMS with an invariant density matrix with full support. In this case the subalgebra $\Acal_d$ has to be replaced by the decoherence-free algebra of the QMS, and the evolution on this algebra is a unitary evolution. Consequently, the definition of the decoherence time still makes sense.

\paragraph{}We shall now briefly explained how we can defined appropriate functional inequalities in order to estimate the decoherence time. Just as in the usual framework of functional inequalities, where the QMS is \emph{primitive} and converges toward its unique faithful invariant state, we can upper bound the trace norm distance, either in terms of the $\chi^2$-divergence, or in terms of the relative entropy between both states. Both define appropriate Lyapunov functionals of the initial state for the evolution, in the sense that they are non-negative and non-increasing along a trajectory, and that they vanish only on states that are equal to their conditional expectation on the decoherence-free algebra. As in the primitive case, the first tentative leads to an upper bound in terms of the spectral gap of a certain symmetrization of the QMS. This spectral gap is also equal to the optimal constant in a generalized form of the Poincaré Inequality (PI). In the case of $\Pcal^\deco$, it has a real spectrum and it is easy to see that the spectral gap is given by $\gamma$. The bound then reads:
\[\norm{\Pcal^\deco_{*t}\left(\rho-E_{\Acal}(\rho)\right)}_{\Tr}\leq\,\sqrt d\,e^{-\gamma t}\,,\]
so that we obtain the following scaling of the decoherence time:
\[\tau_{\chi^2}=\Omega(\log d)\,,\]
where the symbol $\Omega$ means that $\tau_{\chi^2}$ is at least as big as a constant times $\log d$ for large $d$. One can obtain an other estimate by using instead the Pinsker's Inequality to upper bound the trace norm in terms of the relative entropy. We will show how this bound is related to a certain new form of (modified) log-Sobolev Inequality (MLSI), which is equivalent to the exponential decay of the relative entropy. Denoting by $\alpha$ the constant appearing in this inequality, we obtain with this method
\[\norm{\Pcal^\deco_{*t}\left(\rho-E_{\Acal}(\rho)\right)}_{\Tr}\leq\,\sqrt{2\log d}\,e^{-\alpha t}\,.\]
In the case of $\Pcal^\deco$, we shall prove that $2\alpha\geq \gamma$, so that this time we obtain the following scaling of the decoherence time:
\[\tau_{LS}=\Omega(\log (\log d))\,.\]
We see that, as in the case of a primitive QMS and compared with a PI, a MLSI can drastically improve our scaling of the decoherence time with respect to the size of the system. This comes at the price that the log-Sobolev constant $\alpha$ has to be independent of the system size, which is not always the case.

\paragraph{}Therefore, the goal of this article is to introduce new functional inequalities that are relevant even for non-primitive QMS, based on the theory of environment-induced decoherence. We shall define a PI and a modified LSI and prove that they imply the exponential decay of respectively the appropriate notion of variance and relative entropy. As in the primitive case, the PI is equivalent to the spectral gap of a symmetrization of the QMS, which ensured its positivity when the system is finite dimensional. We also prove the positivity of the MLS constant under a natural condition called $\Lbb_1$-regularity of the Dirichlet form. Such condition was only prove for some particular classes of primitive QMS. One on the main contribution of this paper is a proof of the $\Lbb_p$-regularity of the Dirichlet form under a strong form of detailed balance condition \cite{KFGV77,FV07,CM16}. This condition is satisfied by a large class of QMS, including all ones for which strong $\Lbb_p$-regularity was already proved. The proof relies on a new chain rule formula for the Dirichlet form, that generalized the one studied by Carlen and Maas in \cite{CM16}.

\paragraph{}So far, in the literature, the theory of quantum functional inequalities has mainly followed the classical theory and, apart from some typical technical tools required in the non-commutative framework, the main results were obtained in a similar way as in the classical case. However, the generalization to non-primitive QMS displays strong non-classical features such as the presence of quantum correlations (between spatially separated systems) or the existence of quantum coherence ( or Schrödinger-cat's like states), as in the example of the decoherence QMS above. In this article, we present a general framework in which these special features are exposed. Our main contribution and the originality of this work lie in defining the appropriate framework in which these issues can be treated. We formulate the bases of the theory and hope that it will stimulate further work in this direction.

\paragraph{}This article is structured as follows. In Section \ref{sect1}, we recall some results on non-commutative functional inequalities in the primitive case. We also expose the concept of environment-induced decoherence and the corresponding structure induced on the QMS, on which is built the rest of the article. In Section \ref{sect2}, we introduce our functional inequalities in the non-primitive case and prove that they imply the exponential decay of their corresponding Lyapunov functionals. We also prove that the Poincar\'e constant is an upper bound for the modified log-Sobolev constant, similarly to the primitive case. In Section \ref{sect3}, we prove the $\Lbb_p$-regularity of the Dirichlet form under the strong detailed balance condition. We subsequently prove the positivity of our modified log-Sobolev constant under the weak $\Lbb_1$-regularity of the Dirichlet form. We study two particular situations in Section \ref{sect4}. Section \ref{sect5} is dedicated to the study of the decoherence time. We conclude with some remarks in Section \ref{sect6}.

\section{Notations and some preliminaries}\label{sect1}
In this section we introduce our framework and notations and recall the relevant results relative to the primitive case. In Subsection \ref{sect11} we present the formalism of Environment Induced Decoherence (EID) for quantum Markov semigroups (QMS) on finite dimensional Hilbert spaces. The Poincaré Inequality and the modified log-Sobolev Inequality are introduced in the primitive case in Subsection \ref{sect12}. In Subsection \ref{sect13} we recall the structure of the Lindbladian studied in \cite{DFSU1} that we shall use throughout this article.

\subsection{Environment Induced Decoherence for quantum Markov semigroups}\label{sect11}

Let $\Hcal$ be a finite dimensional Hilbert space of dimension $d$. We denote by $\Bcal(\Hcal)$ the Banach space of bounded operators on $\Hcal$ and by $\Bcal_{\text{sa}}(\Hcal)$ the subspace of selfadjoint operators on $\Hcal$, i.e. $\Bcal_{\text{sa}}(\Hcal)=\left\{X=\Bcal(\Hcal);\ X=X^*\right\}$. We write $\Scal(\Hcal)$ the set of positive and trace one operators on $\Hcal$, also called \emph{density matrices}. In the following, we will often identify a density matrix $\rho\in\Scal(\Hcal)$ and the \emph{state} it defines, that is, the positive linear functional $X\in\Bcal(\Hcal)\mapsto\Tr[\rho\, X]$. In particular, we will write $\rho(X)$ for the expected value of $X$ in the state $\rho$. We recall that if $\Hcal=\Hcal_A\otimes\Hcal_B$, where $\Hcal_A$ and $\Hcal_B$ are finite dimensional Hilbert spaces, then the partial trace of a state $\rho_{AB}\in\Scal(\Hcal)$ with respect to $\Hcal_B$ is the unique state $\Tr_{\Hcal_B}[\rho_{AB}]\in\Scal(\Hcal_A)$ such that:
\begin{equation}\label{eq_partialtrace_state}
\Tr\big[\rho_{AB}\,(X\otimes I_{\Hcal_B})\big]=\Tr\big[\Tr_{\Hcal_B}(\rho_{AB})\,X\big],\qquad \forall X\in\Bcal(\Hcal_A)\,.
\end{equation}
Given a state $\rho_B\in\Scal(\Hcal_B)$, the partial trace with respect to $\rho_B$ is the unique linear map $\Tr_{\rho_B}$ from $\Bcal(\Hcal)$ to $\Bcal(\Hcal_A)$ such that for all $X\in\Bcal(\Hcal)$:
\begin{equation}\label{eq_partialtrace_observables}
\Tr\big[X\,(\rho_A\otimes\rho_B)\big]=\Tr\big[\Tr_{\rho_B}(X)\,\rho_A\big],\qquad \forall \rho_A\in\Scal(\Hcal_A)\,.
\end{equation}

Throughout this article we work with a continuous QMS $(\Pcal_t)_{t\geq0}$ acting on $\Bcal(\Hcal)$. Recall that by Lindblad Theorem \cite{Lind}, its generator $\Lcal$ defined by $\Lcal=\underset{t\to0}{\lim}\,\frac1t(\Pcal_t-I)$ and called the \emph{Lindbladian}, takes the form:
\begin{equation}
\Lcal(X)=i[H,X]+\frac{1}{2}\sum_{k\geq1}{\left[L_k^*XL_k-2\left(L_k^*L_kX+XL_k^*L_k\right)\right]}\quad \text{for all }X\in\Bcal(\Hcal)\,,
\label{eqlindblad}
\end{equation}
where $H\in\Bcal_{\text{sa}}(\Hcal)$, where $\{L_k\}_{k\geq1}$ is a finite subset of $\Bcal(\Hcal)$ and where $[\cdot,\cdot]$ denotes the commutator defined as $[X,Y]=XY-XY$ for all $X,Y\in\Bcal(\Hcal)$. $\Pcal$ models the evolution of the observables of the system in the Heisenberg picture of time evolution. In the Schrödinger picture, the evolution of the states is given by the predual $(\Pcal_{*t})_{t\geq0}$ of the QMS for the Hilbert-Schmidt scalar product, that is, the unique trace-preserving QMS such that for all $X,Y\in\Bcal(\Hcal)$,
\[\Tr[\Pcal_t(X)\,Y]=\Tr[X\,\Pcal_{*t}(Y)],\qquad\forall t\geq0\,.\]
Its generator, the predual of $\Lcal$, is the map defined on $\Bcal(\Hcal)$ as:
\begin{equation*}
\Lcal_*(\rho)=-i[H,\rho]+\frac{1}{2}\sum_{k\geq1}{\left[L_k\rho L_k^*-2\left(L_kL_k^*\rho +\rho L_kL_k^*\right)\right]}\quad \text{for all }\rho\in\Scal(\Hcal)\,.
\end{equation*}

We shall always assume that $\Pcal$ admits an invariant state, that is, a density matrix $\sigma\in\Scal(\Hcal)$ such that for all time $t\geq0$, $\sigma(\Pcal_t(X))=\sigma(X)$, for all $X\in\Bcal(\Hcal)$. Equivalently, one has $\Pcal_{*t}(\sigma)=\sigma$ for all $t\geq0$. We shall furthermore assume that $\sigma$ is \emph{faithful}, that is, $\sigma>0$. Under this condition, it was proved for instance in \cite{CSU4} that the maximal algebra on which $\Pcal$ acts as a $*$-automorphism is the \emph{Decoherence-Free Algebra} of $\Pcal$, defined by
\begin{equation}\label{eq_def_DFalgebra}
\Ncal(\Pcal)=\left\{X\in\Bcal(\Hcal),\ \Pcal_t(X^*Y)=\Pcal_t(X)^*\Pcal_t(Y)\text{ and }\Pcal_t(YX^*)=\Pcal_t(Y)\Pcal_t(X)^*\ \forall t\geq0,\,\forall Y\in\Bcal(\Hcal)\right\}\,.
\end{equation}
It means that there exists a continuous one-parameter group of unitary operators $(U_t)_{t\in\R}$ on $\Hcal$ such that for any $X\in\DF$:
\begin{equation}\label{eq_unitary_evolution}
\Pcal_t(X)=U^*_t\,X\,U_t,\qquad\forall t\geq0\,,
\end{equation}
and $\DF$ is the largest subalgebra of $\Bcal(\Hcal)$ such that this holds. We are now ready to state the main result concerning environment-induced decoherence on finite dimensional Hilbert space.

\begin{theo}[Proposition 8 in \cite{CSU3} and Theorem 19 in \cite{CSU4}]\label{theo_deco}
Assume that $\Pcal$ has a faithful invariant state $\sigma$. Then there exists a unique conditional expectation $E_\Ncal$\footnote{In the sense of operator algebra theory, that is, $E_\Ncal$ is a completely positive projection.} from $\Bcal(\Hcal)$ to $\Ncal(\Pcal)$ compatible with $\sigma$ (i.e. $\sigma=\sigma\circ E_{\Ncal}$) and such that
\begin{enumerate}
\item The image of $E_\Ncal$ is the decoherence-free algebra: $\text{Im}\,E_\Ncal=\Ncal(\Pcal)$ and, consequently,
\begin{equation}\label{eq_theo_deco1}
\Bcal(\Hcal)=\DF\oplus\,\text{Ker}\,E_\Ncal\,.
\end{equation}
\item The kernel of $E_\Ncal$ is the subset of observables that vanishes in the long time:
\begin{equation}\label{eq_theo_deco2}
X\in\text{Ker}\,E_\Ncal\qquad\text{iff}\qquad\underset{t\to+\infty}{\lim}\,\Pcal_t(X)=0\,,
\end{equation}
where the convergence is in the uniform topology.
\end{enumerate}
\end{theo}
The decomposition \eqref{eq_theo_deco1} has the following interpretation. The space $\text{Ker}\,E_\Ncal$ is thought as the part of the system which is beyond experimental resolution. Indeed, if the decoherence is fast enough, any measurement of an observable $X\in\text{Ker}\,E_\Ncal$ will give the value $0$. Thus, in the long-time asymptotic, the system behaves effectively like a closed system described by the $*$-algebra $\DF$, on which the evolution is non-dissipative.\\

The conditional expectation $E_\Ncal$ plays a central role in the following. Theorem \ref{theo_deco} can actually be rephrased in the following way: for all $X\in\Bcal(\Hcal)$,
\begin{equation}\label{eq_limit_deco_Hei}
\underset{t\to+\infty}{\lim}\,\Pcal_t\left(X-E_\Ncal(X)\right)=0\,.
\end{equation}
There is an equivalent formulation in the Schrödinger picture similar to Equation \eqref{eq_limitQMSdeco}. First introduce the predual $E_{\Ncal*}$ of the conditional expectation $E_\Ncal$ defined by the relation
\[\Tr[E_\Ncal(X)\,Y]=\Tr[X\,E_{\Ncal*}(Y)],\qquad\forall X,Y\in\Bcal(\Hcal)\,.\]
Then for all state $\rho\in\Scal(\Hcal)$:
\begin{equation}\label{eq_limit_Schro}
\underset{t\to+\infty}{\lim}\,\Pcal_{*t}\left(\rho-E_{\Ncal*}(\rho)\right)=0\,.
\end{equation}
This is the same limit as in Equation \eqref{eq_limitQMSdeco}. Our goal in the following is to adapt the definition of the Poincar\'e Inequality and the modified log-Sobolev Inequality to estimated the speed of this limit through the \emph{decoherence time}
\begin{equation}\label{eq_def_decotime2}
\tau_\deco(\eps)=\min\,\left\{t\geq0\,;\,\norm{\Pcal_{*t}\left(\rho-E_{\Ncal*}(\rho)\right)}_{\Tr}\leq \eps\,;\quad\forall\rho\in\Scal(\Hcal)\right\}\,.
\end{equation}

\begin{ex}\label{ex_deco_decoQMS}
In the example of the \emph{decoherence quantum Markov semigroup}, the maximally mixed state plays the role of a faithful invariant state. The decoherence-free algebra is the algebra of diagonal operators $\Acal_d$ and the conditional expectation $E_\Ncal$ is the projection on this algebra for the Hilbert-Schmidt scalar product.
\end{ex}

We shall now introduce those two functional inequalities in the usual setting.

\subsection{Known result on functional inequalities for primitive quantum Markov semigroups}\label{sect12}

One particular situation is the case where the decoherence-free algebra is trivial, that is, $\DF=\C I_\Hcal$. In this case, the limit in \eqref{eq_theo_deco2} implies that $\sigma$ is the unique invariant state and one has:
\[\underset{t\to+\infty}{\lim}\,\Pcal_t(X)=\Tr[\sigma\,X]I_\Hcal,\qquad\forall X\in\Bcal(\Hcal)\,,\]
which is the usual definition of a \emph{primitive} quantum Markov semigroup \cite{wolftour}. In this case the decoherence time reduced to the well-known \emph{mixing time}:
\[\tau_\text{mixing}(\eps)=\min\,\left\{t\geq0\,;\,\norm{\Pcal_{*t}\left(\rho)-\sigma\right)}_{\Tr}\leq \eps\,;\quad\forall\rho\in\Scal(\Hcal)\right\}\,.
\]
The space $\Bcal(\Hcal)$ can be naturally endowed with an Hilbert space structure with respect to $\sigma$, with scalar product defined for all $X,Y\in\Bcal(\Hcal)$ by:
\begin{equation*}
\sca{X}{Y}_\sigma=\Tr\left[\sigma^{\frac{1}{2}}X^*\sigma^{\frac{1}{2}}Y\right]\,.
\end{equation*}
It defines a norm $\norm{\cdot}_{2,\sigma}$ on $\Bcal(\Hcal)$:
\begin{equation*}
\norm{X}_{2,\sigma}=\Tr\left[\left|\sigma^{\frac{1}{4}}X\sigma^{\frac{1}{4}}\right|^2\right]^{\frac{1}{2}}\,.
\end{equation*}
We denote by $\hat \Pcal$ the adjoint of $\Pcal$ for the scalar product $\sca{\cdot}{\cdot}_\sigma$. It can be checked that $\hat \Pcal$ gives the evolution of the relative density $\sigma^{-\frac{1}{2}}\,\rho\,\sigma^{-\frac{1}{2}}$ of any state $\rho\in\Scal(\Hcal)$, that is,
\[\hat\Pcal_t(\,\sigma^{-\frac{1}{2}}\,\rho\,\sigma^{-\frac{1}{2}}\,)=\sigma^{-\frac{1}{2}}\,\Pcal_t(\rho)\,\sigma^{-\frac{1}{2}}\qquad\forall t\geq0\,.\]
Moreover, the Lindbladian of $\hat \Pcal$ is given by $\hat\Lcal=\sigma^{-\frac{1}{2}}\,\Lcal_*\left(\sigma^{\frac{1}{2}}\,\cdot\,\sigma^{\frac{1}{2}}\right)\,\sigma^{-\frac{1}{2}}$. We say that $\Pcal$ is \emph{reversible} with respect to $\sigma$ if $\Pcal=\hat \Pcal$, or in other words if $\Pcal$ is selfadjoint for $\sca{\cdot}{\cdot}_\sigma$. When studying the $\Lbb_p$-regularity of the QMS, we will use a stronger form of reversibility and we reserve the term \emph{Detailed Balance Condition} with respect to $\sigma$ for this notion.\\ 
When $\Pcal$ is not primitive, there is necessarily not a unique invariant state and it becomes unclear what would be the good notion of reversibility. In the next section, we shall highlight a particular invariant state that we shall choose as a reference state. For short, when the QMS admits a faithful invariant state, we say that $\Pcal$ (or $\Lcal$) is \emph{reversible}, implicitly assuming that it is with respect to this reference state. 

\paragraph{}The main idea behind the use of quantum functional inequalities in the study of the mixing time is the use of appropriate Lyapunov functionals. The idea originated in \cite{DSC96,BT06} in the classical case and was generalized to the quantum setting in \cite{C-M,O-Z,KT2013}. The two Lyapunov functionals we shall be concerned with are the variance and the relative entropy. The \emph{Variance} with respect to $\sigma$ is defined for all $X\in\Bcal(\Hcal)$ as
\begin{equation*}
\Var_\sigma(X)=\norm{X-\sigma(X)I_\Hcal}_{2,\sigma}^2\,.
\end{equation*}
It was proved in \cite{KT2013} that the derivative of the variance along the QMS gives twice the opposite of the \emph{Dirichlet form} of $\Lcal$, defined for all $X\in\Bcal(\Hcal)$ as:
\begin{equation}\label{eq_def_Dirichlet}
\Ecal_\Lcal(X)=-\sca{X}{\Lcal(X)}_\sigma\,.
\end{equation}
We thus have for all time $t\geq0$:
\begin{equation}\label{eq_comp_Dirichlet}
\frac{\partial}{\partial t}\,\Var_\sigma(\Pcal_t(X))=-2\,\Ecal_\Lcal(\Pcal_t(X))\,.
\end{equation}
We emphasize that this differentiation yields the same result even if $\Pcal$ is not primitive and $\sigma$ is any of its faithful invariant states. Consequently, the exponential decay of the variance with decay rate $\lambda>0$ is implied by the following \emph{Poincar\'e Inequality} defined as:
\begin{equation}
\lambda\, \Var_\sigma(X)\leq \Ecal_\Lcal(X)\qquad\text{for all }X\in\Bcal_{\text{sa}}(\Hcal)\,.
\label{eqdefpoinin}
\end{equation}
We denote by $\lambda(\Lcal)$ the optimal constant in this inequality. It can be shown that $\lambda(\Lcal)$ corresponds to the spectral gap of the symmetrized Lindbladian $\frac{\Lcal+\hat\Lcal}{2}$, that is, the absolute value of its greatest non-zero eigenvalue. 

\paragraph{}The \emph{Relative Entropy} with respect to $\sigma$ is defined for all $\rho\in\Scal(\Hcal)$ as
\begin{equation*}
\Dent{\rho}{\sigma}=
\left\{\begin{array}{ll}
& -\Tr\left[\rho\left(\log\rho-\log\sigma\right)\right]\text{ when }\text{supp}\, \rho\subset \text{supp}\,\sigma\,,    \\
&  +\infty\qquad \text{otherwise}\,,
\end{array}\right.
\end{equation*}
where $\operatorname{supp}$ denotes the support of the density matrix. Remark that in our case $\sigma$ is faithful and therefore the relative entropy can not take an infinite value. Spohn defined in \cite{Spo1} the \emph{Entropy Production} as the opposite of its derivative along the evolution:
\begin{equation}\label{eq_def_EP}
\EP_\Lcal(\rho):=-\frac{\partial}{\partial t}\,\Dent{\Pcal_{*t}(\rho)}{\sigma}\geq0\,.
\end{equation}
The entropy production was subsequently computed by Spohn in the same article and is equal to:
\begin{equation}\label{eq_comp_EP}
\EP_\Lcal(\rho)=-\Tr\left[\Lcal_*(\rho)\left(\log\rho-\log\sigma\right)\right]\,.
\end{equation}
This quantity plays a central role in statistical mechanics (see \cite{Spo1,SL78,JPW14} and references therein). The exponential decay of the relative entropy with decay rate $\alpha$ is equivalent to the so-called \emph{modified log-Sobolev Inequality} (MLSI):
\begin{equation}
2\,\alpha\, D\left(\rho||\sigma\right)\leq \EP_\Lcal(\rho)\qquad\text{for all }\rho\in\Scal(\Hcal)\,.
\label{eqdeflogsob1}
\end{equation}
We denote by $\alpha_1(\Lcal)$ the optimal constant in this inequality. Kastoryano and Temme proved in \cite{KT2013} that for a reversible primitive QMS, the MLSI implies the PI as
\begin{equation}\label{eq_comp_MLSI_PI}
\alpha_1(\Lcal)\leq\lambda(\Lcal)\,.
\end{equation}
We shall see that this inequality remains true for non-primitive QMS with the appropriate definitions for each constant.

\paragraph{}In the finite-dimensional case, $\lambda(\Lcal)$ is non-zero if and only if the QMS is primitive, as it is given by the spectral gap of a selfadjoint operator. In the case of the modified log-Sobolev constant the question is still open in the general case. Indeed, it is not even known if the entropy production of a primitive QMS vanishes only for the invariant state, which in view of the MLSI is a necessary condition for $\alpha_1(\Lcal)$ to be positive. Under some regularity assumptions that we shall discuss in Subsection \ref{sect31}, combined results from \cite{TPK} and \cite{KT2013} show that for a reversible and primitive QMS,
\begin{equation}\label{eq_comp_primitive}
\frac{\lambda(\Lcal)}{\log(1/\sigma_{\min})+2}\leq\alpha_2(\Lcal)\leq\alpha_1(\Lcal)\,,
\end{equation}
where $\sigma_{\min}$ is the smallest eigenvalue of $\sigma$ and where $\alpha_2(\Lcal)$ is the usual log-Sobolev constant as defined by Gross in the classical setting. This Inequality asserts that $\alpha_1(\Lcal)>0$, however it lies heavily on Gross' equivalence between the log-Sobolev constant and hypercontractivity. As in the non-primitive case those notions are yet not available we need to relies on different arguments. We shall prove the positivity of $\alpha_1(\Lcal)$ under the regularity conditions mentioned above and without invoking the $\alpha_2$ constant.

\subsection{Structure of the Lindbladian and the invariant states}\label{sect13}

Throughout this article we shall intensively rely on the following structure of the Lindbladian $\Lcal$ and of its invariant states, that was highlighted by the authors in \cite{DFSU1}. Being a $*$-algebra on a finite-dimensional Hilbert space, the DF algebra admits the following structure. Up to a unitary transformation, the Hilbert space $\Hcal$ admits a decomposition as
\begin{equation}
\Hcal=\bigoplus_{i\in I}{\Hcal_i\otimes\Kcal_i}\,,
\label{eqtheostructlind1}
\end{equation}
where $I$ is a finite set, such that $\Ncal(\Pcal)$ is (unitarily isomorphic to) the algebra
\begin{equation}
\DF=\bigoplus_{i\in I}{\Bcal(\Hcal_i)\otimes I_{\Kcal_i}}\,.
\label{eqtheostructlind2}
\end{equation}
The authors in \cite{DFSU1} deduced from this decomposition the following structure of Lindbladian given by Equation \eqref{eqlindblad}. For all $k\geq0$ and all $i\in I$, there exist bounded operators $M_k^{(i)}\in\Bcal(\Kcal_i)$ such that
\begin{equation}
L_k=\bigoplus_{i\in I}{\left(I_{\Hcal_i}\otimes M_k^{(i)}\right)}\,.
\label{eqtheostructlind3}
\end{equation}
For all $i\in I$, there exist bounded operators $M_0^{(i)}\in\Bcal_{\text{sa}}(\Kcal_i)$ and $K_i\in\Bcal_{\text{sa}}(\Hcal_i)$ such that:
\begin{equation}
H=\bigoplus_{i\in I}{\left(K_i\otimes I_{\Kcal_i}+I_{\Hcal_i}\otimes M_0^{(i)}\right)}\,.
\label{eqtheostructlind4}
\end{equation}
Denote by $P_i$ the orthogonal projection on $\Hcal_i\otimes\Kcal_i$. There exist density matrices $\tau_i$ on $\Kcal_i$ such that any invariant state $\sigma$ of $\Pcal$ can be written:
\begin{equation}
\sigma=\bigoplus_{i\in I}\,p_i\,\sigma_i\otimes\tau_i\,,\qquad \Tr[P_i\,\sigma\, P_i]\,,
\label{eqtheoinvstates}
\end{equation}
where $\sigma_i$ is a density matrix in $\Hcal_i$ commuting with $K_i$. Finally we can give explicit formulae for both conditional expectations $E_\Ncal$ and $E_{\Ncal*}$:
\begin{equation}\label{eq_cond_expect}
\begin{aligned}
E_\Ncal(X) &=\bigoplus_{i\in I}\,{\Tr_{\tau_i}[P_i\,X\,P_i]\otimes I_{\Kcal_i}},\qquad \forall X\in\Bcal(\Hcal)\,,\\
E_{\Ncal*}(\rho)& =\bigoplus_{i\in I}\,{\Tr_{\Kcal_i}[P_i\,\rho\, P_i]\otimes\tau_i}\qquad\quad \forall \rho\in\Scal(\Hcal)\,.
\end{aligned}
\end{equation}

From those expressions, we see that the conditional expectation carries more information on the QMS than the decoherence-free algebra. In particular, it does not depend on the choice of the faithful invariant state. Throughout this article, we shall use the following notation: we write for a state $\rho\in\Scal(\Hcal)$:
\begin{equation}\label{eq_notation_stateDF}
\rho_\Ncal=E_{\Ncal*}(\rho)\quad\text{ or equivalently }\quad \rho_\Ncal=\rho\circ E_\Ncal\,.
\end{equation}

\section{Quantum functional inequalities for non-primitive QMS}\label{sect2}

In this section we give the definitions of the Poincar\'e Inequality \eqref{eqdefpoinin} and of the modified log-Sobolev Inequality \eqref{eqdeflogsob1}, which are consistent with the study of environment induced decoherence. The main idea that we wish to transmit is the importance of the reference state in our framework. The choice of the reference state and its properties are presented in Subsection \ref{sect21} below. The PI and the MLSI for not necessarily primitive QMS are introduced in Subsection \ref{sect22} and \ref{sect23} respectively, where we subsequently proved that they induced exponential decay of their respective Lyapunov functionals. We conclude in Subsection \ref{sect24} with the comparison of the two constants we have defined, as in Inequality \eqref{eq_comp_MLSI_PI}.

\subsection{The choice of the reference state}\label{sect21}
In the primitive case, there is a unique invariant state and so it is natural to choose it as the reference state, in particular in the definition of the Dirichlet form and the entropy production. In our case we have to be careful about which invariant state we choose. It appears that an appropriate choice is given by the following density matrix:
\begin{equation}\label{eq_def_sigmatr}
\sigma_{\Tr}=E_{\Ncal*}(\frac{I_\Hcal}{d})\,.
\end{equation}
Equivalently, it is the unique state such that
\[\sigma_{\tr}(E_\Ncal(X))=\frac{1}{d}\Tr[E_\Ncal(X)],\qquad\forall X\in\Bcal(\Hcal)\]
(recall that if $\rho\in\Scal(\Hcal)$, then $\rho(X)=\Tr[\rho\,X]$). This defines a proper faithful state that furthermore has the following central property.

\begin{lem}
With respect to the decomposition of $\Hcal$ given by equations \eqref{eqtheostructlind1} and \eqref{eqtheostructlind2}, $\sigma_{\tr}$ can be written (recall that $P_i$ is the orthogonal projection on $\Hcal_i\otimes\Kcal_i$)
\begin{equation}\label{eq_lem_sigmatr}
\sigma_{\tr}=\sum_{i\in I}{\frac{\Tr[P_i]}{d}I_{\Hcal_i}\otimes\tau_i}\,.
\end{equation}
Consequently, $\DF$ is in the centralizer of $\sigma_{\tr}$. More particularly, for all $X\in\DF$ and all $Y\in\Bcal(\Hcal)$,
\begin{equation}
\Tr[\sigma_{\tr}\,X\,Y]=\Tr[\sigma_{\tr}\, Y\,X]
\label{eqlemdeco}
\end{equation}
( or equivalently $\sigma_{\tr}(XY)=\sigma_{\tr}(YX)$).
\label{lemdeco}
\end{lem}

\begin{proof}
Equation  \eqref{eq_lem_sigmatr} is just Equation \eqref{eq_cond_expect} applied to the maximally mixed state. Using the form of the elements of $\DF$ given by Equation \eqref{eqtheostructlind2}, we see that they commute with $\sigma_{\tr}$ which implies the second part of the lemma. 
\end{proof}

Other useful properties of $\sigma_{\tr}$ are listed below.

\begin{lem}\ 
\begin{enumerate}
\item For all $X\in\Bcal(\Hcal)$,
\begin{equation}\label{eq_lem_cond}
\sigma_{\Tr}^{\frac12}\,E_\Ncal(X)\,\sigma_{\Tr}^{\frac12}=E_{\Ncal*}(\sigma_{\Tr}^{\frac12}\,X\,\sigma_{\Tr}^{\frac12})\,,
\end{equation}
or equivalently $\sigma_{\Tr}^{-\frac12}\,E_{\Ncal*}(\rho)\,\sigma_{\Tr}^{-\frac12}=E_\Ncal(\sigma_{\Tr}^{-\frac12}\,\rho\,\sigma_{\Tr}^{-\frac12})$ for all $\rho\in\Scal(\Hcal)$.
\item For all $X,Y\in\Bcal(\Hcal)$
\begin{equation}\label{eqlemcond}
\sca{E_\Ncal(X)}{Y}_{\sigma_{\tr}}=\sca{E_\Ncal(X)}{E_\Ncal(Y)}_{\sigma_{\tr}}=\sca{X}{E_\Ncal(Y)}_{\sigma_{\tr}}\,.
\end{equation}
In particular, $E_\Ncal$ is the orthogonal projection on $\Ncal(\Pcal)$ with respect to the scalar product $\sca{\cdot}{\cdot}_{\sigma_{\tr}}$ and $\text{Ker}\ E_\Ncal=\Ncal(\Pcal)^\perp$.
\end{enumerate}
\label{lemcond}
\end{lem}

\begin{proof}\ 
\begin{enumerate}
\item First notice that as a property of conditionnal expectations, $E_\Ncal(X\,Y\,Z)=X\,E_\Ncal(Y)\,Z$ for all $X,Z\in\DF$ and $Y\in\Bcal(\Hcal)$. Consequently, using the notations introduced in Equations \eqref{eqtheostructlind1} and \eqref{eqtheostructlind2},
\begin{equation*}
\sigma_{\Tr}^{\frac12}\,E_\Ncal(X)\,\sigma_{\Tr}^{\frac12} = \sum_{i\in I}\,\sigma_{\Tr}^{\frac12}\,E_\Ncal(P_i\,X\,P_i)\,\sigma_{\Tr}^{\frac12}\,.
\end{equation*}
We can study each term in the sum separetly. Direct computations using the definition of the partial traces defined in Equations \eqref{eq_partialtrace_state} and \eqref{eq_partialtrace_observables} shows that for all $i\in I$,
\[\Tr_{\tau_i}[P_i\,X\,P_i]=\Tr_{\Kcal_i}[(I_{\Hcal_i}\otimes\tau_i^{1/2})\,(P_i\,X\,P_i)(I_{\Hcal_i}\otimes\tau_i^{1/2})]\,,\]
and so
\begin{align*}
\sigma_{\Tr}^{\frac12}\,E_\Ncal(P_i\,X\,P_i)\,\sigma_{\Tr}^{\frac12}
& = \frac{\Tr[P_i]}{d}\,\Tr_{\tau_i}[P_i\,X\,P_i]\otimes\tau_i \\
& = \frac{\Tr[P_i]}{d}\,\Tr_{\Kcal_i}[(I_{\Hcal_i}\otimes\tau_i^{1/2})\,(P_i\,X\,P_i)(I_{\Hcal_i}\otimes\tau_i^{1/2})]\otimes\tau_i \\
& = E_{\Ncal*}\left(\sigma_{\Tr}^{\frac12}\,P_i\,X\,P_i\,\sigma_{\Tr}^{\frac12}\right)\,.
\end{align*}
\item As $E_\Ncal$ is a conditional expectation, we have $E_\Ncal^2=E_\Ncal$, which can also directly be checked from Equation \eqref{eq_cond_expect}. Consequently we only need to check that it is selfadjoint with respect to $\sca{\cdot}{\cdot}_{\sigma_{\Tr}}$. Using the first part of this lemma, we obtain
\begin{align*}
\sca{E_\Ncal(X)}{Y}_{\sigma_{\tr}}
& = \Tr\left[\sigma_{\tr}^{\frac{1}{2}}\,E_\Ncal(X^*)\,\sigma_{\tr}^{\frac{1}{2}}\,Y\right] \\
& = \Tr\left[E_{\Ncal*}(\sigma_{\Tr}^{\frac12}\,X^*\,\sigma_{\Tr}^{\frac12})\,Y\right] \\
& = \Tr\left[\sigma_{\Tr}^{\frac12}\,X^*\sigma_{\Tr}^{\frac12}\,E_{\Ncal*}(Y)\right] \\
& = \sca{X}{E_\Ncal(Y)}_{\sigma_{Tr}}\,.
\end{align*}
We obtain Equation (\ref{eqlemcond}).
\end{enumerate}
\end{proof}

\begin{rem}
Equation \eqref{eqlemcond} is still true if one replaces the scalar product $\sca{\cdot}{\cdot}_{\sigma_{\tr}}$ by either one of the following two:
\begin{align*}
& (X,Y)\mapsto \sca{X}{Y}_{HS}=\Tr[X^*\, Y]\,, \\
& (X,Y)\mapsto \Tr[\sigma_{\tr}\, X^*\, Y]\,.
\end{align*}
For the latter, it will even stay true for any other faithful invariant state. However, the fact that Equation \eqref{eqlemcond} holds for $\sca{\cdot}{\cdot}_{\sigma_{\tr}}$ is specific to the choice of $\sigma_{\tr}$ as a reference state.
\end{rem}

\begin{de}\label{def_reversible}
We say that a QMS $\Pcal$ is \emph{reversible} if it admits a faithful invariant state, so that $\sigma_{\tr}$ is well-defined, and if it is selfadjoint with respect to $\sca{\cdot}{\cdot}_{\sigma_{\tr}}$.
\end{de}

From now on, whenever there is no ambiguity, we will forget the subscript $\sigma_{\tr}$ in the definitions of the $\Lbb_2$ norm and the scalar product.

\subsection{The Decoherence-Free Variance and Poincar\'e Inequality}\label{sect22}
We now describe a generalization of the Poincar\'e Inequality \eqref{eqdefpoinin}. In the case where $\Pcal$ is a primitive QMS, the variance of an observable $X\in\Bcal(\Hcal)$ in the state $\sigma$ has a nice geometric interpretation. It represents the square of the norm of the orthogonal projection of $X$ on $\left(\C I_\Hcal\right)^\perp$, with respect to the scalar product $\sca{\cdot}{\cdot}_\sigma$. That is,
\[\Var_\sigma(X)=\norm{\text{Proj}_{(\C I)^\perp}\ (X)}^2_{2,\sigma}\,.\]
We want to keep this geometric interpretation of the variance. By Lemma \ref{lemcond}, the conditional expectation $E_\Ncal$ is the orthogonal projection on $\Ncal(\Pcal)$ with respect to the scalar product defined by $\sigma_{\tr}$. As a consequence the following definition appears as the logical analogue of the traditional variance.
\begin{de}
We define the \emph{Decoherence-Free Variance} (DF-variance) for all $X\in\Bcal(\Hcal)$, as the square of the norm of the projection of $X$ onto the orthogonal of $\Ncal(\Pcal)$, that is:
\begin{equation}
\Var_{\Ncal}(X)=\norm{X-E_\Ncal(X)}^2_{2}\,.
\label{eqdefDFvariance}
\end{equation}
\label{deDFvariance}
\end{de}
Note that we get the usual definition when $\DF=\C I$, as in this case $E_\Ncal(X)=\sigma(X)I_\Hcal$: the DF-variance is indeed a generalization of the variance.
The following lemma emphasizes the particular choice of $\sigma_{\tr}$ as a reference state.

\begin{lem}
For all $X\in\Bcal(\Hcal)$, one has
\begin{equation}\label{eq_lem_DFvariance}
\Var_{\Ncal}(X)=\Var_{\sigma_{\tr}}(X)-\Var_{\sigma_{\tr}}(E_\Ncal(X))
\end{equation}
and
\begin{equation}
\Var_{\sigma_{\tr}}(E_\Ncal(X))=\sigma_{\tr}\left(E_\Ncal(X)^2\right)-\sigma_{\tr}(X)^2\,.
\label{eqlemDFvariance}
\end{equation}
\label{lemDFvariance}
\end{lem}

\begin{rem}\label{rem_DFvariance}
Notice that consequently, because of Equation \eqref{eq_lem_DFvariance}, one has $\Var_\Ncal(X)\leq\Var_{\sigma_{\tr}}(X)$ for all $X\in\Bcal(\Hcal)$.
\end{rem}

\begin{proof}
We write
\[X-\sigma_{\tr}(X)I_\Hcal=\left(X-E_\Ncal(X)\right)+\left(E_\Ncal(X)-\sigma_{\tr}(X)I_\Hcal\right)\,.\]
As, by Lemma \ref{lemcond}, $X-E_\Ncal(X)$ is in the orthogonal of $\DF$ for the concern scalar product, in particular it is orthogonal to $E_\Ncal(X)-\sigma_{\tr}(X)I_\Hcal$. Equation \eqref{eq_lem_DFvariance} follows directly.\\
In order to prove Equation \eqref{eqlemDFvariance}, we use the fact that 
\[\Var_{\sigma_{\tr}}\left(E_\Ncal(X)\right)=\norm{E_\Ncal(X)}_2^2-\sigma_{\tr}\left(E_\Ncal(X)\right)^2=\norm{E_\Ncal(X)}_2^2-\sigma_{\tr}\left(X\right)^2\,,\]
as it holds that $\sigma_{\tr}\circ E_\Ncal=\sigma_{\tr}$. By Lemma \ref{lemdeco}, $\sigma_{\tr}$ is tracial on $\DF$ and we have $\norm{E_\Ncal(X)}_2^2=\sigma_{\tr}\left(E_\Ncal(X)^2\right)$.
\end{proof}

One can now defined a generalization of the Poincar\'e Inequality \eqref{eqdefpoinin} with respect to the DF-variance. We say that $\Pcal$ satisfies a \emph{Decoherence-free Poincar\'e Inequality} with constant $\lambda>0$ ( and we write $\text{PI}_\Ncal(\lambda)$) if the following inequality holds for all $X\in\Bcal_{\text{sa}}(\Hcal)$:
\begin{equation}
\lambda\,\Var_{\Ncal}(X)\leq \Ecal_\Lcal(X)\,.
\label{eqtheospeedEIDvar1}
\end{equation}
where $\Ecal_\Lcal$ is defined as in Equation \eqref{eq_def_Dirichlet} with $\sigma=\sigma_{\tr}$. We write $\lambda(\Lcal)$ the best constant which satisfies this inequality.

\begin{theo}\label{theo_speedEIDvar}
If $\text{PI}_\Ncal(\lambda)$ holds then for all $X\in\Bcal_{\text{sa}}(\Hcal)$, one has an exponential decay of the DF-variance with exponential rate $2\lambda$:
\begin{equation}
\Var_{\Ncal}(\Pcal_t(X))\leq e^{-2\,\lambda\,t}\Var_{\Ncal}(X)\,.
\label{eqtheospeedEIDvar2}
\end{equation}
Moreover, $\lambda(\Lcal)$ coincide with the spectral gap of $\frac{\Lcal+\hat \Lcal}{2}$, which reduces to the one of $\Lcal$ when $\Pcal$ is reversible.
\end{theo}

\begin{proof}
The proof is the same as for the usual Poincar\'e Inequality. In order to prove the first part of the Proposition we compute the derivative of $\Var_{\Ncal}(\Pcal_t(X))$. Assume that $\text{PI}_\Ncal(\lambda)$ holds for some $\lambda>0$. First, by Equation \eqref{eq_comp_Dirichlet}, we have
\[\frac{\partial}{\partial t}\Var_{\sigma_{\tr}}(\Pcal_t(X))\, =\, -2\,\Ecal_\Lcal(X)\, \leq\, -2\,\lambda\, \Var_{\Ncal}(\Pcal_t(X))\,.\]
We now exploit Equation \eqref{eqlemDFvariance} to prove that
\[\frac{\partial}{\partial t}\,\Var_{\sigma_{\tr}}\left(E_\Ncal(X)\right)=0\,.\]
First remark that $\Pcal$ and $E_\Ncal$ commute. Indeed, for all $X\in\Bcal(\Hcal)$ and all $t\geq0$, $\Pcal_t(E_\Ncal(X))\in\Bcal(\Hcal)$ and $\Pcal_t(X-E_\Ncal(X))\in\Ker\,E_\Ncal$, so that $E_\Ncal(\Pcal_t(X)))=E_\Ncal(\Pcal_t(E_\Ncal(X))=\Pcal_t(E_\Ncal(X))$. Consequently,
\begin{align*}
\sigma_{\tr}\left(E_\Ncal(\Pcal_t(X))^2\right)
 =\sigma_{\tr}\left(\Pcal_t[E_\Ncal(X)]^2\right) %\\
 =\sigma_{\tr}\left(\Pcal_t[E_\Ncal(X)^2]\right) %\\
=\sigma_{\tr}\left(E_\Ncal(X)^2\right)\,,
\end{align*}
where in the second line we use that $\Pcal$ is a unitary evolution on $\DF$ and in the last line the fact that $\sigma_{\tr}$ is an invariant state. This latter property, combined with Equation \eqref{eqlemDFvariance}, shows that $\frac{\partial}{\partial t}\Var_{\sigma_{\tr}}(E_\Ncal(\Pcal_t(X)))=0$. Consequently,
\[\frac{\partial}{\partial t}\Var_{\Ncal}(\Pcal_t(X)) \leq  -2\,\lambda\, \Var_{\Ncal}(\Pcal_t(X))\,.\]
Equation \eqref{eqtheospeedEIDvar2} follows by integrating this inequality.\\
Now, if $\Lcal$ is reversible it can be diagonalized in some orthonormal basis. We write $\lambda_1,...,\lambda_n$ its eigenvalues associated to the orthogonal projections $Q_0,Q_1,...,Q_n$ and such that $0=\lambda_0\leq-\lambda_1\leq\cdots\leq-\lambda_n$. As proved in \cite{CSU3}, $\Ncal(\Pcal)$ coincide with the kernel of $\Lcal$, so that $Q_0$ is the orthogonal projection on $\Ncal(\Pcal)$ (that is, $Q_0=E_\Ncal$). Then, for all $X\in\Bcal(\Hcal)$:
\begin{align*}
\Ecal_\Lcal(X)
 = \sum_{k=1}^n{-\lambda_k\,|\sca{Q_k(X)}{X}|^2} %\\
 \geq -\lambda_1\,\norm{X-Q_0(X)}_2^2 \geq\lambda_1\, \norm{X-E_\Ncal(X)}_{2}^2%\\
 = -\lambda_1\, \Var_{\Ncal}(X)\,.
\end{align*}
As equality is achieved for any eigenvectors associated to the eigenvalue $\lambda_1$, the claim is proved for reversible $\Lcal$. In the general case, observe that for selfadjoint operators $X\in\Bcal(\Hcal)$, the Dirichlet form associated to $\Lcal$ and $\hat\Lcal$ are the same, so that they also coincide with the Dirichlet form of $\frac{\Lcal+\hat \Lcal}{2}$: $\Ecal_{\Lcal}=\Ecal_{\hat \Lcal}=\Ecal_{\frac{\Lcal+\hat\Lcal}{2}}$ on $\Bcal_{\text{sa}}(\Hcal)$. This concludes the proof.
\end{proof}

\begin{ex}[The decoherence quantum Markov semigroup]
Recall that, for $\gamma\in\R_*^+$, the decoherence QMS $\Pcal^\deco_{*}$ is defined for all states $\rho\in\Scal(\Hcal)$ by:
\[\Pcal^\deco_{*t}(\rho)=e^{-\gamma\,t}\,\rho+(1-e^{-\gamma t})\,E_{\Ncal*}(\rho)\,,\]
where $E_\Ncal=E_{\Ncal*}$ is the orthogonal projection on the algebra of diagonal operators for the Hilbert-Schmidt scalar product. The Lindbladian of this QMS in the Schr\"odinger picture is given by
\begin{equation}\label{eq_ex_decoQMS1}
\Lcal_*^\deco(\rho)=\gamma\left(E_{\Ncal*}(\rho)-\rho\right)\,.
\end{equation}
Clearly, in this example $\sigma_{\tr}$ is the maximally-mixed state $\frac{I_\Hcal}{d}$. We can compute explicitly the Dirichlet form: for all $X\in\Bcal(\Hcal)$,
\begin{equation*}
\Ecal_{\Lcal^\deco}(X) = -\gamma\,\sca{X}{E_\Ncal(X)-X}_{\sigma_{\tr}} = \gamma\,\Var_\Ncal(X)\,.
\end{equation*}
Consequently, we directly obtain that $\lambda(\Lcal^\deco)=\gamma$. Of course, this can also be directly checked from the definition of $\Lcal^\deco$.
\end{ex}

As a simple corollary, we directly get that the Dirichlet form vanishes only on the DF algebra.

\begin{coro}\label{coro_DFalgebra}
We have $\Ecal_\Lcal(X)=0$ if and only if $X\in\DF$.
\end{coro}

\begin{proof}
Indeed, as $\Hcal$ is finite-dimensional, the spectral gap is positive. Now, let $X\in\Bcal(\Hcal)$ be such that $\Ecal_\Lcal(X)=0$. Then the DF Poincar\'e Inequality implies that $\Var_\Ncal(X)=0$, so that by definition $X=E_\Ncal(X)$ and $X\in\DF$. The other implication comes from $\Ecal_\Lcal=\Ecal_{\frac{\Lcal+\hat \Lcal}{2}}$ together with the fact already mentioned above that $\Ncal(e^{t(\frac{\Lcal+\hat \Lcal}{2})})=\Ker \frac{\Lcal+\hat \Lcal}{2}$.
\end{proof}

\subsection{The Decoherence-Free Entropy and log-Sobolev Inequality}\label{sect23}

We shall now propose a generalization of the modified log-Sobolev Inequality \eqref{eqdeflogsob1}. We recall the notation introduced in \eqref{eq_notation_stateDF}: for a state $\rho\in\Scal(\Hcal)$, we write $\rho_\Ncal=\rho\circ E_\Ncal$.

\begin{de}\label{deDFE1}
We define the \emph{Decoherence-Free relative entropy} (DF-relative entropy) for all $\rho\in\Scal(\Hcal)$ as
\begin{equation}\label{eqdeDFE1}
\DentDF{\rho}:=\Dent{\rho}{\rho_\Ncal}\,.
\end{equation}
\end{de}
The DF-relative entropy represents the information lost in the environment during the decoherence process. It also reduces to the usual relative entropy $\Dent{\rho}{\sigma}$ in the case where $\DF=\C I$, as in this case $\rho_\Ncal=\sigma$, the unique invariant state. As for Lemma \ref{lemDFvariance}, the specific choice of $\sigma_{\tr}$ as a reference state is motivated by the following lemma.

\begin{lem}\label{lem_DFentropy}
For all state $\rho\in\Scal(\Hcal)$, one has:
\begin{equation}
\DentDF{\rho}=\Dent{\rho}{\sigma_{\tr}}-\Dent{\rho_\Ncal}{\sigma_{\tr}}\,.
\label{eqlemDFentropy}
\end{equation}
\end{lem}

\begin{rem}
Consequently we have $\DentDF{\rho}\leq\Dent{\rho}{\sigma_{\tr}}$.
\end{rem}

\begin{proof}
We use the notations introduce in Equations \eqref{eqtheostructlind1} and below. Recall that, for some finite dimensional Hilbert space $\Kcal$, if $X\in\Bcal(\Hcal)$ and $Y\in\Bcal(\Kcal)$ are positive semi-definite operator, then $\log \left(X\otimes Y\right)=\left(\log X\right)\otimes I_\Kcal+I_\Hcal\otimes \left(\log Y\right)$, where the logarithm is defined on their support. Besides, if $\Hcal=\Kcal$ and $X$ and $Y$ have orthogonal support, then $\log(X+Y)=\log X+\log Y$. We shall adopt the following notation: we write $\rho_{\Hcal_i}=\Tr_{\Kcal_i}[P_i\,\rho\,P_i]$ and $\rho_{\Kcal_i}=\Tr_{\Hcal_i}[P_i\,\rho\,P_i]$. With these notations, $\rho_\Ncal=\sum_{i}\,\rho_{\Hcal_i}\otimes \tau_i$. We start by computing the left-hand side of Equation \eqref{eqlemDFentropy}:
\begin{align*}
\DentDF{\rho}
& =\Tr\left[\rho\,\left(\log\rho-\log\rho_\Ncal\right)\right] \\
& =\Tr\left[\rho\log\rho\right]-\sum_{i\in I}\,\Tr\big[(P_i\,\rho\,P_i)\,\left(\log(\rho_{\Hcal_i})\otimes I_{\Kcal_i}+I_{\Hcal_i}\otimes\log(\tau_i)\right)\big]\\
& =\Tr\left[\rho\log\rho\right]-\sum_{i\in I}\,\Tr\big[\rho_{\Hcal_i}\log\rho_{\Hcal_i}\big]-\sum_{i\in I}\,\Tr\big[\rho_{\Kcal_i}\log\tau_i\big]\,.
\end{align*}
Now we compute the first term of the right-hand side. We write $N_i=\dim\,\Hcal_i$.
\begin{align*}
\Dent{\rho}{\sigma_{\tr}}
& = \Tr\big[\rho\big(\log\rho-\log\sigma_{\tr}\big)\big] \\
& = \Tr\left[\rho\log\rho\right]-\sum_{i\in I}\,\Tr\left[(P_i\,\rho\,P_i)\left(\log (\frac{I_{\Hcal_i}}{N})\otimes I_{\Kcal_i}+I_{\Hcal_i}\otimes\log(\tau_i)\right)\right]\\
& = \Tr\left[\rho\log\rho\right] + \sum_{i\in I}\,\Tr[P_i\,\rho\,P_i]\log N_i -\sum_{i\in I}\,\Tr\big[\rho_{\Kcal_i}\log\tau_i\big] \\
& = \DentDF{\rho} + \sum_{i\in I}\,\Tr\big[\rho_{\Hcal_i}\left(\log\rho_{\Hcal_i}+\log N_i\right)\big]\,.
\end{align*}
As $\Tr_{\Kcal_i}[P_i\,\rho_\Ncal\,P_i]=\rho_{\Hcal_i}$, the same computation applied to $\rho_\Ncal$ shows that $\Dent{\rho_\Ncal}{\sigma_{\tr}}=\sum_{i\in I}\,\Tr\big[\rho_{\Hcal_i}\left(\log\rho_{\Hcal_i}+\log N_i\right)\big]$, which concludes the proof.
\end{proof}
One can now define a generalization of the log-Sobolev Inequality leading to the exponential decay of the DF-relative entropy. We say that $\Pcal$ satisfies a \emph{modified Decoherence-Free log-Sobolev Inequality} with constant $\alpha>0$ (and we write $\text{MLSI}_\Ncal(\alpha)$) if for all state $\rho\in\Scal(\Hcal)$, the following inequality holds:
\begin{equation}
2\,\alpha\,\DentDF{\rho}\,\leq \,\EP_\Lcal(\rho)\,,
\label{eqtheoDFentropy1}
\end{equation}
where $\EP_\Lcal$ is defined as in Equation \eqref{eq_comp_EP} with $\sigma=\sigma_{\tr}$. We write $\alpha_\Ncal(\Lcal)$ the best constant in the previous inequality:
\begin{equation}\label{eq_def_bestconstantLSI}
\alpha_\Ncal(\Lcal)=\underset{\rho\in\Scal(\Hcal)}{\inf}\, \frac{\EP_\Lcal(\rho)}{2\,\DentDF{\rho}}\,.
\end{equation}

\begin{theo}\label{theoDFentropy}
$\text{MLSI}_\Ncal(\alpha)$ holds for some $\alpha>0$ if and only if for all $\rho\in\Scal(\Hcal)$ and all $t\geq0$,
\begin{equation}
\DentDF{\Pcal_{*t}(\rho)}\leq e^{-2\,\alpha\, t}\,\DentDF{\rho}\,.
\label{eqtheoDFentropy2}
\end{equation}
\end{theo}

\begin{rem}\label{rem_DF_MLSI}
Compared to the case of the DF-Poincaré Inequality, the DF-modified log-Sobolev Inequality is equivalent to the exponential decay of the DF-relative entropy. This is explained by the presence of the factor two in the differentiation of the variance, which does not appear in the PI Inequality (see the proof below). 
\end{rem}

\begin{proof}
Spohn's computation of the entropy production given in Equation \eqref{eq_comp_EP} gives:
\[\frac{\partial}{\partial t}\Dent{\Pcal_t(\rho)}{\sigma_{\tr}}=-\EP_\Lcal(\rho)\,.\]
Consequently, to prove that $\text{MLSI}_\Ncal(\alpha)$ implies the exponential decay of the DF relative entropy,  we just need to prove that $\frac{\partial}{\partial t}\Dent{\Pcal_{*t}(\rho_\Ncal)}{\sigma_{\tr}}=0$. This is indeed the case, as
\[\Dent{\Pcal_{*t}(\rho_\Ncal)}{\sigma_{\tr}}=\Dent{U_t\,\rho_\Ncal\,U_t^*}{U_t\,\sigma_{\tr}\,U_t^*}=\Dent{\rho_\Ncal}{\sigma_{\tr}}\,,\]
where $U_t$ is defined in Equation \eqref{eq_unitary_evolution} and where we use the invariance of the relative entropy under simultaneous unitary conjugation of both states.\\
In order to prove that $\text{MLSI}_\Ncal(\alpha)$ is a necessary condition, remark that Inequality \eqref{eqtheoDFentropy2} is assumed to hold for all $t\geq0$ and that equality holds a $t=0$, so that the result follows by differentiating this inequality at $t=0$.
\end{proof}

\begin{ex}[The decoherence QMS]
A simple computation shows that, for the decoherence QMS,
\begin{equation}\label{eq_ex_decoQMS2}
\EP_{\Lcal^\deco}(\rho)=\gamma\left(\Dent{\rho}{\rho_\Ncal}+\Dent{\rho_\Ncal}{\rho}\right)\,.
\end{equation}
Consequently, we obtain that
\[\alpha_\Ncal(\Lcal^\deco)\geq\frac{\gamma}{2}\geq\frac{\lambda(\Lcal^\deco)}{2}\,.\]
This in particular shows that $\alpha_\Ncal(\Lcal)>0$.
\end{ex}

\subsection{Comparison between the constants}\label{sect24}

We now show an extension to the non-primitive case of the well-known comparison between the modified log-Sobolev constant and the spectral gap.

\begin{theo}\label{theocomparconst}
Let $\Pcal$ be a reversible QMS on $\Bcal(\Hcal)$, with generator $\Lcal$. Then the DF-log-Sobolev constant $\alpha_\Ncal(\Lcal)$ and the spectral gap $\lambda(\Lcal)$ satisfy:
\begin{equation}
\alpha_\Ncal(\Lcal)\leq \lambda(\Lcal)\,.
\label{eqtheocomparconst}
\end{equation}
\end{theo}

We begin by recalling a result whose proof can be found in the two articles \cite{C-M,KT2013}. We give a sketch of the proof for sake of completeness.

\begin{lem}\label{lem_estimation}
Let $\sigma\in\Scal(\Hcal)$ be a faithful invariant state of $\Pcal$ and define, for $\eps>0$, $\rho_\eps=\sigma+\eps\,\sigma_{\tr}^{\frac{1}{2}}\,Y\,{\sigma_{\tr}}^{\frac{1}{2}}$, where $Y\in\Bcal_{\text{sa}}(\Hcal)$ is such that $E_\Ncal(Y)=0$. Then, for $\eps$ small enough, $\rho_\eps$ defines a faithful state and:
\begin{equation}\label{eq_proof_primitive}
\begin{aligned}
&\Dent{\rho_\eps}{\sigma}=\frac{\eps^2}{2}\,\sca{Y}{\Theta_{\sigma}(\sigma_{\tr}^{\frac12}\,Y\,\sigma_{\tr}^{\frac12})}_{\sigma_{\Tr}}+ o (\eps^2)\,,   \\
& \EP_\Lcal(\rho_\eps)=\eps^2\,\Ecal_\Ncal\left(Y\,,\,\Theta_{\sigma}(\sigma_{\tr}^{\frac12}\,Y\,\sigma_{\tr}^{\frac12})\right)+ o (\eps^2)\,, 
\end{aligned}
\end{equation}
where $\Theta_{\sigma}$ is the map $Z\in\Bcal(\Hcal)\mapsto\Theta_{\sigma}(Z)=\int_{0}^{+\infty}\,(t+{\sigma})^{-1}\,Z\,(t+{\sigma})^{-1}\,dt$. Furthermore, $\Theta_{\sigma}$ is a CP map which is positive definite with respect to the weight scalar product $\sca{\cdot}{\cdot}_{\sigma}$.
\end{lem}

\begin{rem}\label{rem_inversion}\ 
\begin{itemize}
\item The map $Y\mapsto \Theta_{\sigma}(\sigma_{\tr}^{\frac12}\,Y\,\sigma_{\tr}^{\frac12})$ is a signature of the non-commutative nature of the system. Indeed, it maps the identity operator on itself and so will be the case for any operator commuting with $\sigma$. In particular, it is the trivial map for $\sigma=\frac{I_{\Hcal}}{d}$.
\item The map $\Theta_\sigma$ is a particular instance of \emph{inversion of $\sigma$} as defined for instance in \cite{TKRW}. There are many different choices of inversions, all related to different monotone Riemmanian metrics on $\Scal(\Hcal)$. This particular inversion is related to the relative entropy (see \cite{TKRW} and reference therein).
\end{itemize}
\end{rem}

\begin{proof}
Both expansions relies on the integral representation of the logarithm of a semi-definite operator:
\[\log\,X=\int_0^{+\infty}\,\frac1t-\frac{1}{t+X}\,dt\,.\]
Combined with the operator identity $X^{-1}-Y^{-1}=X^{-1}(Y-X)Y^{-1}$, this gives
\[\log\,X-\log\,Y=\int_0^{+\infty}\,\frac{1}{t+Y}(Y-X)\frac{1}{t+X}\,dt\,.\]
Applying this expression with the expansion of $(t+\rho_\eps)^{-1}$ and writing $g=\sigma_{\tr}^{\frac{1}{2}}\,Y\,{\sigma_{\tr}}^{\frac{1}{2}}$, we get
\begin{align*}
\log\,\rho_\eps-\log\sigma
& =\eps\int_0^{+\infty}\,\frac{1}{t+\sigma}\,g\,\frac{1}{t+\sigma}\,dt 
-\eps^2\int_0^{+\infty}\,\frac{1}{t+\sigma}\,g\,\frac{1}{t+\sigma}\,g\,\frac{1}{t+\sigma}\,dt+\mathcal O(\eps^3) \\
& = \eps\,\Theta_\sigma(g)-\eps^2\int_0^{+\infty}\,\frac{1}{t+\sigma}\,g\,\frac{1}{t+\sigma}\,g\,\frac{1}{t+\sigma}\,dt+\mathcal O(\eps^3)\,.
\end{align*}
Taking the expectation of this operator under the state $\rho_\eps$ now gives
\begin{align*}
\Dent{\rho_\eps}{\sigma}
 = \eps\,\Tr[\sigma\,\Theta_\sigma(g)] %\\
 -\eps^2\,\Tr\big[\sigma\,\int_0^{+\infty}\,\frac{1}{t+\sigma}\,g\,\frac{1}{t+\sigma}\,g\,\frac{1}{t+\sigma}\,dt\big]  %\\
 + \eps^2\,\Tr\big[g\,\Theta_\sigma(g)\big]+\mathcal O(\eps^3)+\mathcal O(\eps^3)\,.
\end{align*}
Evaluating the first term in the right-hand side of the previous expression in an orthonormal basis in which $\sigma$ is diagonal, we get $\Tr[\sigma\,\Theta_\sigma(g)]=\Tr[g]$. Recall that $\sigma_{\tr}^{\frac12}\,E_\Ncal(Y)\,\sigma_{\tr}^{\frac12}=E_{\Ncal*}\left(\sigma_{\tr}^{\frac12}\,Y\,\sigma_{\tr}^{\frac12}\right)$, so that, as $E_{\Ncal}(Y)=0$, $\Tr[g]=\Tr[E_{\Ncal_*}(\sigma_{\tr}^{\frac12}\,Y\,\sigma_{\tr}^{\frac12})]=0$. As in \cite{KT2013}, the two other terms can be evaluated similarly and we get
\[\Dent{\rho_\eps}{\sigma}=\frac{\eps^2}{2}\,\Tr\big[g\,\Theta_\sigma(g)\big]+\mathcal O(\eps^3)\,,\]
which leads to the desired identity when replacing $g$ by $\sigma_{\tr}^{\frac12}\,Y\,\sigma_{\tr}^{\frac12}$. The second order expansion of $\EP_\Lcal(\rho_\eps)$ is obtained similarly.
\end{proof}

We can now proceed to the proof of our theorem.

\begin{proof}[Proof of Theorem \ref{theocomparconst}]
We follow the usual proof, which in the quantum primitive case can be found in \cite{C-M,KT2013}. Let $Y\in\Bcal_{\text{sa}}(\Hcal)$ be such that $Y\in\Ker\,E_\Ncal$. Then, for $\eps$ small enough, $\rho_\eps=\sigma_{\tr}+\sigma_{\tr}^{\frac{1}{2}}\,Y\,\sigma_{\tr}^{\frac{1}{2}}$ is a state and the estimations in \eqref{eq_proof_primitive} give:
\begin{align*}
& \EP_\Lcal(\rho_\eps)=\eps^2\,\Ecal_\Lcal\left(Y\,,\,\Theta_{\sigma_{\tr}}(\sigma_{\tr}^{\frac{1}{2}}\,Y\,\sigma_{\tr}^{\frac{1}{2}})\right)+ o (\eps^2)\,,  \\
&\Dent{\rho_\eps}{\sigma_{\tr}}=\frac{\eps^2}{2}\,\sca{Y}{\Theta_{\sigma_{\tr}}(\sigma_{\tr}^{\frac{1}{2}}\,Y\,\sigma_{\tr}^{\frac{1}{2}})}_{\sigma_{\tr}}+ o (\eps^2)\,,
\end{align*}
as $\Lcal=\hat\Lcal$ as it is reversible. Remark that, as we supposed that $E_\Ncal(Y)=0$ and therefore $E_{\Ncal*}(\sigma_{\tr}^{\frac{1}{2}}\,Y\,\sigma_{\tr}^{\frac{1}{2}})=0$, we have $\Dent{\rho_\eps}{\sigma_{\tr}}=\DentDF{\rho_\eps}$. Consequently, applying the DF log-Sobolev Inequality to $\rho_\eps$ and taking the limit $\eps\to0$ gives:
\[\alpha_\Ncal(\Lcal)\,\sca{Y}{\Theta_{\sigma_{\tr}}(Y)}_{\sigma_{\tr}}\leq\,\Ecal_\Lcal\left(Y\,,\,\Lcal\circ\Theta_{\sigma_{\tr}}(Y)\right)\,.\]
The argument in either \cite{C-M} or \cite{KT2013} allows to remove the operator $\Theta_{\sigma_{\tr}}$ in this inequality, giving:
\[\alpha_\Ncal(\Lcal)\,\norm{Y}^2_{\sigma_{\tr}}\leq\,\Ecal_{\Lcal}(Y)\,.\]
As $E_\Ncal(Y)=0$, $\norm{Y}^2_{\sigma_{\tr}}=\Var_\Ncal(Y)$. For general $X\in\Bcal_{\text{sa}}(\Hcal)$ and as $\Ecal_\Lcal(X)=\Ecal_\Lcal(X-E_\Ncal(X))$, we can just replace $X$ by $X-E_\Ncal(X)$ in the previous inequality. We obtain that $\alpha_\Ncal(\Lcal)$ satisfies the Poincaré Inequality \eqref{eqtheospeedEIDvar1}, so that it is upper bounded by $\lambda(\Lcal)$.
\end{proof}

\begin{rem}\label{rem_theo_comp}
As in the primitive case, Theorem \ref{theocomparconst} still holds without the reversibility assumption in the case where the QMS is doubly stochastic. This comes from the fact that in this case, $\sigma_{\tr}$ is the maximally mixed state so that $\Theta_{\sigma_{\tr}}$ is the trivial map (see the first point of Remark \ref{rem_inversion}).
\end{rem}

\section{Positivity of the modified log-Sobolev constant and comparison between the constants}\label{sect3}

Compared to the spectral gap, it is not straightforward that the DF log-Sobolev constant is non-zero. Even in the primitive case, to ensure this property one need Inequality \eqref{eq_comp_primitive} and thus an additional regularity assumption on the Dirichlet form: the quantum Markov semigroup is called \emph{$\mathbb L_1$-regular} if for all $\rho\in\Scal(\Hcal)$,
\begin{equation}\label{eq_regularity}
\EP_\Lcal(\rho)\geq\,2\,\Ecal_\Lcal\left(\sigma_{\tr}^{-\frac14}\,\rho^{\frac12}\,\sigma_{\tr}^{-\frac14}\right)\,.
\end{equation}

One way to realize that $\alpha_\Ncal(\Lcal)$ is not trivially positive is to remark that, compared to the Dirichlet form and Lemma \ref{coro_DFalgebra}, it is not clear whether the entropy production vanishes if and only if the state undergoes a reversible evolution. That is, we do not know if $\EP_\Lcal(\rho)=0$ if and only if $\rho=E_{\Ncal*}(\rho)$. As important in physics as this result seems, we could not find a general statement on this in the literature (see \cite{SL78} for a proof in the case of Davies Lindbladians). Let define properly this property.
\begin{de}\label{de_prop_EP}
We call \eqref{DFEP} for \emph{Entropy Production Condition} the following property of the entropy production:
\begin{equation}\label{DFEP}\tag{EPC}
\EP(\rho)=0\qquad\text{ iff }\qquad\rho=E_{\Ncal*}(\rho)\,.
\end{equation}
\end{de}
Using Corollary \ref{coro_DFalgebra} and the fact that $\rho=E_{\Ncal*}(\rho)$ if and only if $\sigma_{\tr}^{-\frac12}\,\rho\,\sigma_{\tr}^{-\frac12}\in\Ncal(\Pcal)$, we directly obtain:
\begin{prop}\label{prop_EP}
Assume that $\Pcal$ is $\mathbb L_1$-regular. Then Property \eqref{DFEP} holds.
\end{prop}

We shall prove in Subsection \ref{sect32} that the DF modified log-Sobolev constant is positive whenever the Dirichlet form is weakly $\mathbb L_1$-regular. Before that, we prove in Subsection \ref{sect31} that QMS with "strong" detailed balance condition are strongly $\Lbb_p$-regular.

\subsection{$\Lbb_p$-regularity of the Dirichlet form under strong detailed balance condition}\label{sect31}

Recall the definition from \cite{O-Z,KT2013} of the $p$-Dirichlet form, with $p>1$:
\[\Ecal_{p,\Lcal}(X)=-\frac{p}{2(p-1)}\,\sca{\sigma_{\tr}^{-\frac{1}{2q}}\left|\sigma_{\tr}^{\frac{1}{2p}}\,X\,\sigma_{\tr}^{\frac{1}{2p}}\right|^{p/q}\sigma_{\tr}^{-\frac{1}{2q}}}{\Lcal(X)}_{\sigma_{\tr}}\,.\]
This definition is extended to $p=!$ by considering the limit $p\to1$:
\[\Ecal_{1,\Lcal}(X)=-\frac12\Tr\left[\sigma_{\tr}^{\frac12}\,\Lcal(X)\,\sigma_{\tr}^{\frac12}\left(\log(\sigma_{\tr}^{\frac12}\,X\,\sigma_{\tr}^{\frac12})-\log\sigma_{\tr}\right)\right]\,.\]
Taking $X=\sigma_{\tr}^{-\frac12}\,\rho\,\sigma_{\tr}^{-\frac12}$, we recognize the entropy production:
\[\EP_\Lcal(\rho)=2\,\Ecal_{1,\Lcal}(X)\,.\]
Strong regularity of the Dirichlet form was defined in \cite{O-Z} and latter in \cite{KT2013,C-M} in the following way. We say that the Dirichlet form is \emph{strongly $\Lbb_p$-regular} if for all $p\geq1$ and all $X\in\Bcal_{\text{sa}}(\Hcal)$:
\begin{equation}\label{eq_strong_regularity}
\Ecal_{p,\Lcal}(X)\geq\frac{2}{p}\,\Ecal_{2,\Lcal}\left(\sigma_{\tr}^{-\frac{1}{4}}\left|\sigma_{\tr}^{\frac{1}{2p}}\,X\,\sigma_{\tr}^{\frac{1}{2p}}\right|^{p/2}\sigma_{\tr}^{-\frac{1}{4}}\right)\,.
\end{equation}
The same authors also defines a notion of weak regularity, which is expected to hold when the QMS is not reversible. We say that the Dirichlet form is \emph{weakly $\Lbb_p$-regular} if for all $p\geq1$ and all $X\in\Bcal_{\text{sa}}(\Hcal)$:
\begin{equation}\label{eq_weak_regularity}
\Ecal_{p,\Lcal}(X)\geq\left\lbrace\begin{array}{ccc}
\Ecal_{2,\Lcal}\left(\sigma_{\tr}^{-\frac{1}{4}}\left|\sigma_{\tr}^{\frac{1}{2p}}\,X\,\sigma_{\tr}^{\frac{1}{2p}}\right|^{p/2}\sigma_{\tr}^{-\frac{1}{4}}\right),&\qquad& 1\leq p \leq2\,, \\
(p-1)\,\Ecal_{2,\Lcal}\left(\sigma_{\tr}^{-\frac{1}{4}}\left|\sigma_{\tr}^{\frac{1}{2p}}\,X\,\sigma_{\tr}^{\frac{1}{2p}}\right|^{p/2}\sigma_{\tr}^{-\frac{1}{4}}\right),&\qquad& p\geq2\,.
\end{array}\right.
\end{equation}
Clearly, as $\Ecal_{2,\Lcal}=\Ecal_{\Lcal}$, we see that the $\mathbb L_1$-regularity of $\Pcal$ given by \eqref{eq_regularity} is just the restriction of the weak regularity to the case $p=1$. Strong and weak regularity of the Dirichlet form have already been proved for several classes of primitive QMS, such as doubly stochastic QMS (that is, unital and trace-preserving QMS) \cite{C-M,KT2013,O-Z}, Davies QMS (that arose in the weak-coupling limit) or Liouvillian Lindbladian (of the form $I-\hat T T$ for a unital CP map $T$) \cite{KT2013}. In \cite{KT2013}, the authors conjectured that weak regularity always holds for primitive QMS, and that strong regularity holds if furthermore the QMS is reversible. 

\paragraph{}We prove here the strong $\Lbb_p$ regularity for a large class of QMS. This class actually includes all the examples quoted above. Our proof is based on a representation of the Dirichlet form in terms of a derivation, as studied for instance in \cite{CS03}. We believe that our strategy is somehow more natural than the ones previously used and could shade some light on the general case. Our assumption is the following.
\begin{de}\label{def_DBC}
Let $\Pcal$ be a QMS with Lindbladian $\Lcal$ and with a faithful invariant state. We say that $\Pcal$ has the $\sigma_{\tr}$-DBC property whenever for all $X,Y\in\Bcal(\Hcal)$ and all time $t\geq0$,
\begin{equation}\label{eq_def_DBC}
\Tr[\sigma_{\tr}\,X^*\Pcal_t(Y)]=\Tr[\sigma_{\tr}\,\Pcal_t(X^*)Y]\,,
\end{equation}
that is, whenever $\Pcal$ is selfadjoint with respect to the scalar product $(X,Y)\mapsto \Tr[\sigma_{\tr}\,X^*Y]$.
\end{de}
The $\sigma_{\tr}$-DBC property implies the reversibility of the QMS and so is a stronger condition. It also implies strong constraints on the structure of the Lindbladian, as exposed below.

\paragraph{}We can now state our result.
\begin{theo}\label{theo_Lp_regularity}
Assume that the QMS $\Pcal$ possesses the $\sigma_{\tr}$-DBC property. Then $\Pcal$ is strongly $\Lbb_p$-regular for all $p\geq1$, that is, for all positive semi-definite $X\in\Bcal(\Hcal)$ and all $p\geq1$, Inequality \eqref{eq_strong_regularity} holds.
\end{theo}

If the QMS satisfies the $\sigma_{\tr}$-DBC condition, it is proved in \cite{CM16} that the Dirichlet form can be written in the following way:
\begin{equation}\label{eq_Dirichlet_DBC}
\Ecal_\Lcal(X,Y)=\sum_{j\in J}\,\sca{\partial_j(X)}{\partial_j(Y)}_{\sigma_{\tr}}\,,
\end{equation}
where $J$ is a finite set and where $\partial_j$ are derivations on $\Bcal(\Hcal)$, so that there exist $V_j\in\Bcal(\Hcal)$ such that
\[\partial_j\,:X\in\Bcal(\Hcal)\,\mapsto\,\partial_j(X)=[V_j,X]\,.\]
Furthermore, the $V_j$ are eigenvectors of the modular operator associated to $\sigma_{\tr}$, which translates to
\[\sigma_{\tr}\,V_j=e^{\omega_j}\,V_j\,\sigma_{\tr}\,,\qquad \omega_j\in\R\,.\]
This implies that for all $p\in\R\backslash\{0\}$,
\begin{equation}\label{eq_valpropre_modulargroupe}
\sigma_{\tr}^{\frac1{2p}}V_j=e^{\omega_j/2p}\,V_j\,\sigma_{\tr}^{\frac1{2p}}\,.
\end{equation}

We shall now built a non-commutative functional calculus associated to the derivations $\partial_j$. This functional calculus can either be seen as a generalization of the one developed in \cite{CS03} (see also \cite{GIS}), in the context of non-commutative Dirichlet form on $\Ccal^*$-algebras, or as a generalization of the chain rule formula in \cite{CM16}. It is based on the following representation of the space of continuous functions on some interval $I$ on $\Bcal(\Hcal)$.\\
\\We denote by $C(I)$ the Banach space of continuous, complex valued functions on $I$, where $I$ is a closed and bounded interval. In our case, $I$ will be a bounded interval containing the spectrum of semi-definite operators. The Banach space $C(I\times I)$ becomes a $*$-algebra when endowed with the involution $f^*:(x,y)\mapsto \overline{f(y,x)}$, $f\in C(I\times I)$. We endowed $\Kcal:=\Bcal(\Hcal)$ with an Hilbert space structure inherited from the Hilbert-Schmidt scalar product:
\[(X,Y)\in\Bcal(\Hcal)\times\Bcal(\Hcal)\mapsto \scal{X}{Y}_{HS}=\Tr[X^*\, Y]\,.\]
For any couple of elements $X,Y\in\Bcal_{\text{sa}}(\Hcal)$, we denote by $\pi_{X,Y}$ the $*$-representation of the $\Ccal^*$ algebra $C(I\times I)$ into $\Bcal(\Kcal)$, where $I$ is a bounded interval containing $\Sp(X)$ and $\Sp(Y)$, uniquely determined by
\[\pi_{X,Y}(f\otimes g)\cdot Z=
f(X)\,Z\, g(Y)\,,\qquad f,g\in C(I),\,Z\in\Kcal\,.
\]
We collect the properties of this $*$-representation in the following lemma. The proof consists just of simple verifications so we omit it.
\begin{lem}\label{lem_functcal}$\pi_{X,Y}$ is indeed a $*$-representation between $\Ccal^*$-algebras, that is:
\begin{enumerate}
\item $\pi_{X,Y}(\Ind)=I_\Kcal$, where $\Ind$ is the constant function on $I\times I$ equal to $1$.
\item $\pi_{X,Y}(f^*g)=\pi_{X,Y}(f)^*\pi_{X,Y}(g)$ for all $f,g\in C(I\times I)$.
\item If $f\in\Ccal(I\times I)$ is a non-negative function, then $\pi_{X,Y}(f)$ is a positive semi-definite operator on $\Kcal$ for the Hilbert-Schmidt scalar product.
\end{enumerate}
\end{lem}
For any function $f\in C(I)$, we write $\tilde f$ the function in $C(I\times I)$ defined by
\[
\tilde f(s,t)=
\left\lbrace
\begin{array}{ccc}
\frac{f(s)-f(t)}{s-t}\quad&\text{if }&s\ne t\,, \\
f'(s)\quad&\text{if }&s=t\,. 
\end{array}\right.\]
The following lemma gives a generalization of the chain rule formula associated to a derivation. This is the central tool in our proof of the strong $\Lbb_p$-regularity.

\begin{lem}\label{lem_NC_funct_calculus}
Let $f\in\Ccal_0^1(I)$ be such that $f(0)=0$, where $\Ccal_0^1(I)$ is the space of derivable functions on the bounded interval $I$ with continuous (and therefore bounded) derivative. For all $V\in\Bcal(\Hcal)$ and all selfadjoint operators $X,Y\in\Bcal_{\text{sa}}(\Hcal)$ such that $\text{sp}(X),\text{sp}(Y)\subset I$, we have
\begin{equation}\label{eq_lem_NC_funct_calculus}
Vf(Y)-f(X)V=\pi_{X,Y}(\tilde f)\cdot(VY-XV)\,.
\end{equation}
\end{lem}

\begin{proof}
We first prove Equation \eqref{eq_lem_NC_funct_calculus} for $f:x\mapsto x^n$, where $n\in\Nbb$ is a fixed integer. We use a recursive argument: we have
\begin{align*}
VY^{n+1}-X^{n+1}V
& = \left(VY^n-X^nV\right)Y+X^n\left(VY-XV\right) \\
& = \left(\sum_{k=0}^{n-1}\,X^k(VY-XV)Y^{n-1-k}\right)Y+X^n\left(VY-XV\right) \\
& = \sum_{k=0}^n\,X^k(VY-XV)Y^{n-k}  \\
& = \pi_{X,Y}(\tilde f)\cdot(VY-XV)\,.
\end{align*}
By linearity of the $*$-representation $\pi_{X,Y}$, this also holds for all polynomials $P$ such that $P(0)=0$. We want to apply the Stone-Weierstrass Theorem to conclude. To this end, we endow $\Ccal_0^1(I)$ with a Banach space structure with the norm
\[\norm{f}_{\Ccal_0^1(I)}=\norm{f'}_\infty\,,\qquad f\in\Ccal_0^1(I)\,.\]
Remark that for all $f\in\Ccal_0^1(I)$,
\[\norm{f}_{\Ccal_0^1(I)}=\underset{x,y\in I}{\sup}\,\big\{\frac{f(x)-f(y)}{x-y}\big\}=\norm{\tilde f}_\infty\,.\]
This implies that $\pi_{X,Y}(\tilde f_n)\to \pi_{X,Y}(\tilde f)$ whenever $f,f_n\in\Ccal_0^1(I)$ and where $(f_n)$ is a sequence converging to $f$ in the Banach space $\Ccal_0^1(I)$. This concludes the proof.
\end{proof}

\begin{rem}\ 
\begin{itemize}
\item When $X=Y$, Equation \eqref{eq_lem_NC_funct_calculus} is just the well-known chain rule formula for derivations as studied in \cite{CS03,GIS}. We should simply write $\pi_X$ instead of $\pi_{X,Y}$ in this case.
\item The case where $f(x,y)=\frac{\log x-\log y}{x-y}$ and where $X=e^{-\omega/2}Z$ and $Y=e^{w/2}Z$ for an observable $Z\in\Bcal_{\text{sa}}(\Hcal)$ was studied in \cite{CM16}. However, the authors did not highlight the role of the $*$-representation.
\end{itemize}
\end{rem}

We are now ready to prove Theorem \ref{theo_Lp_regularity}.

\begin{proof}[Proof of Theorem \ref{theo_Lp_regularity}]
From Equation \eqref{eq_valpropre_modulargroupe} it can be directly checked that the two following expressions holds for all $X\in\Bcal(\Hcal)$, all $j\in J$ and all $q,p\geq1$:
%\begin{equation}
\begin{align}
& \partial_j\left(I_{q,p}(X)\right)
= \Gamma_{\sigma_{\Tr}}^{-\frac1q}\left(V_j\,\Gamma_{\sigma_{\Tr}}^{\frac1p}(e^{-w_j/2p}X)^{\frac{p}{q}}-\Gamma_{\sigma_{\Tr}}^{\frac1p}(e^{w_j/2p}X)^{\frac{p}{q}}\,V_j\right)\,, \label{eq_commutation_Vj_1}\\
& \partial_j(X)
=\Gamma_{\sigma_{\Tr}}^{-\frac1p}\left(V_j\,(e^{-w_j/2p}X)-(e^{w_j/2p}X)\,V_j\right)  \label{eq_commutation_Vj_2}\,,
\end{align}
where $\Gamma_{\sigma_{\tr}}^r(X)=\sigma_{\tr}^{\frac r2}\,X\,\sigma_{\tr}^{\frac r2}$ and where $I_{q,p}(X)=\sigma_{\tr}^{-\frac1{2q}}\left(\sigma_{\tr}^{\frac1{2p}}\,X\,\sigma_{\tr}^{\frac1{2p}}\right)^{\frac pq}\sigma_{\tr}^{-\frac1{2q}}$ (remark that the last equation is just the first one applied to $q=p$).
Let $X$ be a positive semi-definite operator on $\Hcal$. We write $Y=\Gamma_{\sigma_{\Tr}}^{\frac{1}{p}}(X)$. We also write $f:x\in\R^+\mapsto x^{p/2}$ and $g:x\in\R^+\mapsto x^{p-1}$. We first exploit the expression of the Dirichlet form in terms of the derivations given by Equation \eqref{eq_Dirichlet_DBC} and Equation \eqref{eq_commutation_Vj_1} to obtain:
\begin{align*}
\Ecal_\Lcal(I_{2,p}(X))
& = \sum_{j\in J}\,\norm{\Gamma_{\sigma_{\Tr}}^{-\frac12}\left(V_j\,(e^{-w_j/2p}Y)^{\frac p2}-(e^{w_j/2p}Y)^{\frac p2}\,V_j\right)}_{2,\sigma_{\Tr}}^2 \\
& = \sum_{j\in J}\,\norm{ \pi_{e^{w_j/2p}Y,e^{-w_j/2p}Y}(\tilde f)\cdot\left(V_j\,(e^{-w_j/2p}Y)-(e^{w_j/2p}Y)V_j\right)}^2_{2,HS} \\
& = \sum_{j\in J}\,\scal{V_j\,(e^{-w_j/2p}Y)-(e^{w_j/2p}Y)V_j}{ \pi_{e^{w_j/2p}Y,e^{-w_j/2p}Y}(\tilde f^2)\cdot\left(V_j\,(e^{-w_j/2p}Y)-(e^{w_j/2p}Y)V_j\right) }_{HS}\,,
\end{align*}
where we use that $\pi_{X,Y}$ is a $*$-representation and that $\tilde f$ is symmetric non-negative, so that $\pi_{X,Y}(\tilde f)^*\pi_{X,Y}(\tilde f)=\pi_{X,Y}(\tilde f^2)$. Recall the well know Inequality (see \cite{Bak94} Proposition 3.1 for instance) for all $s,t\in\R^+$ and all $p\geq1$:
\[\left(s^{p/2}-t^{p/2}\right)^2\leq\frac{p^2}{4(p-1)}\,(s^{p-1}-t^{p-1})(s-t)\,.\]
This inequality can be reinterpreted in terms of $\tilde f$ and $\tilde g$ as:
\[\tilde f^2\leq \frac{p^2}{4(p-1)}\,\tilde g\,.\]
Consequently, by point 3. of Lemma \ref{lem_functcal}, $\pi_{e^{w_j/2p}Y,e^{-w_j/2p}Y}(\tilde f^2)\leq \frac{p^2}{4(p-1)}\pi_{e^{w_j/2p}Y,e^{-w_j/2p}Y}(\tilde g)$ and
\begin{align*}
& \Ecal_\Lcal(I_{2,p}(X)) \\
&\qquad \leq \frac{p^2}{4(p-1)}\,\sum_{j\in J}\,\scal{V_j\,(e^{-w_j/2p}Y)-(e^{w_j/2p}Y)V_j}{\pi_{e^{w_j/2p}Y,e^{-w_j/2p}Y}(\tilde g)\cdot\left(V_j\,(e^{-w_j/2p}Y)-(e^{w_j/2p}Y)V_j\right)}_{HS}\\
&\qquad = \frac{p^2}{4(p-1)}\,\sum_{j\in J}\,\scal{V_j\,(e^{-w_j/2p}Y)-(e^{w_j/2p}Y)V_j}{V_j\,(e^{-w_j/2p}Y)^{p-1}-(e^{w_j/2p}Y)^{p-1}V_j}_{HS}\\
&\qquad = \frac{p^2}{4(p-1)}\,\sum_{j\in J}\,\sca{\partial_j\left(I_{q,p}(X)\right)}{\partial_j (X)}_{\sigma_{\Tr}}=\frac p2\,\Ecal_{p,\Lcal}(X)\,,
\end{align*}
where in the second line we used the chain rule formula of Equation \eqref{eq_lem_NC_funct_calculus} and where in third line we used Equations \eqref{eq_commutation_Vj_1} and \eqref{eq_commutation_Vj_2}. This concludes the proof.
\end{proof}

Taking the limit $p\to1$ in Equation \eqref{eq_strong_regularity} with $X=\sigma_{\tr}^{-\frac12}\,\rho\,\sigma_{\tr}^{-\frac12}$, we obtain the desired strong $\Lbb_1$-regularity.

\begin{coro}\label{coro_regularity}
Let $\Pcal$ be a $\sigma$-DBC QMS on $\Bcal(\Hcal)$. Then for all state $\rho\in\Scal(\Hcal)$,
\begin{equation}
\EP_\Lcal(\rho)\geq\,4\,\Ecal_\Lcal\left(\sigma_{\tr}^{-\frac14}\,\rho^{\frac12}\,\sigma_{\tr}^{-\frac14}\right)\,.
\end{equation}
In particular, property \eqref{DFEP} holds.
\end{coro}

\subsection{Positivity of the DF log-Sobolev constant}\label{sect32}

We now state our main theorem concerning the positivity of the DF log-Sobolev constant.

\begin{theo}\label{theo_DFlogSob}
Let $\Pcal$ be a quantum Markov semigroup on $\Bcal(\Hcal)$, with generator $\Lcal$ and with a faithful invariant state. Assume that its Dirichlet form is $\Lbb_1$-regular, as defined in Equation \eqref{eq_regularity}. Then the modified DF log-Sobolev constant $\alpha_\Ncal(\Lcal)$ is strictly positive:
\begin{equation*}
\alpha_\Ncal(\Lcal)>0\,.
\end{equation*}
\end{theo}

\begin{proof}
We write $\Ncal_*(\Pcal)$ the image of $E_{\Ncal*}$. That is, $\sigma\in\Ncal_*(\Pcal)$ iff $E_{\Ncal*}(\sigma)=\sigma$. By definition, the modified DF log-Sobolev constant is given by
\[\alpha_\Ncal(\Lcal)=\underset{\rho\in\Scal(\Hcal)\,\backslash\,\Ncal_*(\Pcal)}{\inf}\,\frac{\EP_\Lcal(\rho)}{2\,\Dent{\rho}{\sigma}}\,.\]
We proceed by contradiction and we assume that $\alpha_\Ncal(\Lcal)=0$. Then, by Property \eqref{DFEP} (which is induced by the $\Lbb_1$ regularity) and by compactness of $\Scal(\Hcal)$, it means that for all $\eta>0$, there exists a state $\rho\in\Scal(\Hcal)\,\backslash\,\Ncal_*(\Pcal)$ such that $\norm{\rho-E_{\Ncal*}(\rho)}_1\leq\eta$ and
\[\frac{\EP_\Lcal(\rho)}{2\,\DentDF{\rho}}\leq\eta\,.\]
We claim that this is a contradiction. Indeed, we shall prove the following fact: there exists $\eta>0$ and $C>0$ such that for all $\rho\in\Scal(\Hcal)$, $\norm{\rho-E_{\Ncal*}(\rho)}_1\leq\eta$ implies that 
\[\frac{\EP_\Lcal(\rho)}{2\,\DentDF{\rho_\eps}}>C\,.\]
More precisely, we only need to show that for all faithful $\sigma\in\Ncal_*(\Pcal)$, there exists $\eta>0$ and $C>0$ depending on $\sigma$ such that for all $Y\in\Bcal_{\text{sa}}(\Hcal)\,\backslash\,\DF$ with $E_\Ncal(Y)=0$, defining $\rho_\eps=\sigma+\eps\,\sigma_{\tr}^{\frac12}\,Y\,\sigma_{\tr}^{\frac12}$ as in Lemma \ref{lem_estimation}, we have for all $\eps\leq\eta$:
\begin{equation}\label{eq_proof_raybis}
\frac{\EP_\Lcal(\rho_\eps)}{2\,\DentDF{\rho_\eps}}>C\,.
\end{equation}
Indeed, invoking again the compactness of $\Ncal_*(\Hcal)$, $\eta$ and $C$ can be chosen independently of $\sigma$. The restriction to faithful states is not a restriction since the set of invertible operators is dense in $\Scal(\Hcal)$. Then, for any faithful $\rho\in\Scal(\Hcal)\,\backslash\,\Ncal_*(\Pcal)$ such that $\norm{\rho-E_{\Ncal*}(\rho)}_1\leq\eta$, we can set $\sigma=E_{\Ncal*}(\rho)$ and $Y=\frac{1}{\eta}\left(\sigma^{-\frac12}\,\rho\,\sigma^{-\frac12}-I_\Hcal\right)$, so that $E_\Ncal(Y)=0$ and $\rho=\sigma+\eta\,\sigma_{\tr}^{\frac12}\,Y\,\sigma_{\tr}^{\frac12}$. Consequently, Inequality \eqref{eq_proof_raybis} holds, a contradiction with the fact that $\alpha_\Ncal(\Lcal)=0$.\\
We now prove our claim. Take a faithful $\sigma\in\Ncal_*(\Pcal)$, $Y\in\Bcal_{\text{sa}}(\Hcal)\backslash\Ncal_*(\Pcal)$ with $E_\Ncal(Y)=0$ and define $\rho_\eps$ as above. By Lemma \ref{lem_estimation},
\begin{equation}\label{eq_proof_ray}
\frac{\EP_\Lcal(\rho_\eps)}{2\,\DentDF{\rho_\eps}}\overset{\eps\to0}{\longrightarrow}\frac{\Ecal_\Lcal\left(Y\,,\,\Theta_{\sigma}(\sigma_{\tr}^{\frac12}\,Y\,\sigma_{\tr}^{\frac12})\right)}{\sca{Y}{\Theta_{\sigma}(\sigma_{\tr}^{\frac12}\,Y\,\sigma_{\tr}^{\frac12})}_{\sigma_{\Tr}}}\,.
\end{equation}
Both the numerator and the denominator in the right-hand side are the first possibly non-zero term in the Taylor expansion of a positive quantity, so that there are both non-negative. The denominator is furthermore easily seen to be positive as it is equal to $\norm{\Theta_{\sigma}^{1/2}(\sigma_{\tr}^{\frac12}\,Y\,\sigma_{\tr}^{\frac12})}_{2,\sigma_{\Tr}}$, and $\Theta_\sigma^{1/2}$ is a positive definite operator. Thus, it remains to prove that $\Ecal_\Lcal\left(Y\,,\,\Theta_{\sigma}(\sigma_{\tr}^{\frac12}\,Y\,\sigma_{\tr}^{\frac12})\right)>0$. In order to prove it, we shall expand the right-hand side of the $\Lbb_1$-regularity condition \eqref{eq_regularity} up to second order, with $\rho=\rho_\eps$, starting with the expansion of $\rho_\eps^{\frac12}$. Recall the integral representation
\[\sqrt x=\frac1\pi\,\int_0^{+\infty}\sqrt t\,(\frac 1t-\frac {1}{t+x})\,dt\,,\]
which together with the operator identity $X^{-1}-Y^{-1}=X^{-1}(X-Y)Y^{-1}$ gives
\[\rho_\eps^{\frac12}=\sigma^{\frac12}-\frac{\eps}{\pi}\,\int_0^{+\infty}\sqrt t\, \left(\frac{1}{t+\sigma}\,\sigma_{\tr}^{\frac12}\,Y\,\sigma_{\tr}^{\frac12}\,\frac{1}{t+\sigma}\right)dt+\mathcal O(\eps^2)\,.\]
Remark that, as $\frac{\sqrt x-\sqrt y}{x-y}=\frac{1}{\pi}\int_0^{+\infty}\frac{\sqrt t\,dt}{(t+x)(t+y)}$, we have
\[\frac1\pi\int_0^{+\infty}\sqrt t\, \left(\frac{1}{t+\sigma}\,\sigma_{\tr}^{\frac12}\,Y\,\sigma_{\tr}^{\frac12}\,\frac{1}{t+\sigma}\right)dt
=\pi_\sigma(\tilde f)\cdot(\sigma_{\tr}^{\frac12}\,Y\,\sigma_{\tr}^{\frac12})\,,\]
where $f:x\mapsto \sqrt x$ and where $\pi_\sigma$ is the $*$-representation introduced in Subsection \ref{sect31}. Write $\sigma=X\sigma_{\Tr}$, where $X\in\DF$ is necessarily a positive definite operator as $\sigma>0$. We thus have $\sigma_{\tr}^{-\frac14}\,\sigma^{\frac12}\,\sigma_{\tr}^{-\frac14}=X$. Then, using that $\Ecal_\Lcal(X,Z)=\Ecal_\Lcal(Z,X)=0$ for all $Z\in\Bcal(\Hcal)$ as $X\in\DF$, we get:
\[\Ecal_\Lcal(\sigma_{\tr}^{-\frac14}\,\rho_\eps^{\frac12}\,\sigma_{\tr}^{-\frac14})=\eps^2\,\Ecal_\Lcal\left(\sigma_{\tr}^{-\frac14}\,\big[\pi_\sigma(\tilde f)\cdot(\sigma_{\tr}^{\frac12}\,Y\,\sigma_{\tr}^{\frac12})\big]\,\sigma_{\tr}^{-\frac14}\right)+\mathcal O(\eps^3)\,.\]
By the $\Lbb_1$ regularity \eqref{eq_regularity} and Equation \eqref{eq_proof_primitive} in Lemma \ref{lem_estimation}, we get
\[2\,\Ecal_\Lcal\left(\sigma_{\tr}^{-\frac14}\,\big[\pi_\sigma(\tilde f)\cdot(\sigma_{\tr}^{\frac12}\,Y\,\sigma_{\tr}^{\frac12})\big]\,\sigma_{\tr}^{-\frac14}\right)
\leq\Ecal_\Lcal\left(Y\,,\,\Theta_{\sigma}(\sigma_{\tr}^{\frac12}\,Y\,\sigma_{\tr}^{\frac12})\right)\,.\]
Define $Z=\sigma_{\tr}^{-\frac14}\,\big[\pi_\sigma(\tilde f)\cdot(\sigma_{\tr}^{\frac12}\,Y\,\sigma_{\tr}^{\frac12})\big]\,\sigma_{\tr}^{-\frac14}$. If the right-hand side of the above inequality is equal to $0$, by Corollary \ref{coro_DFalgebra} we know that $Z\in\DF$. Furthermore, we have
\[\sigma_{\tr}^{\frac12}\,Y\,\sigma_{\tr}^{\frac12}=\pi_\sigma(\frac{1}{\tilde f})\cdot(\sigma_{\tr}^{\frac14}\,Z\,\sigma_{\tr}^{\frac14})\,.\]
As $(\frac1{\tilde f})(x,y)=\sqrt x+\sqrt y$ and using that $\sigma_{\tr}$ commutes with $X,Z\in\DF$, we finally get $Y=ZX^{\frac12}+X^{\frac12}Z\in\DF$, a contradiction. We thus obtain that $\Ecal_\Lcal\left(Y\,,\,\Theta_{\sigma}(\sigma_{\tr}^{\frac12}\,Y\,\sigma_{\tr}^{\frac12})\right)>0$. The proof is concluded by taking for instance
\[C=\frac{1}{2}\,\frac{\Ecal_\Lcal\left(Y\,,\,\Theta_{\sigma}(\sigma_{\tr}^{\frac12}\,Y\,\sigma_{\tr}^{\frac12})\right)}{\sca{Y}{\Theta_{\sigma}(Y)}_{\sigma}}>0\,,\]
so that there exists $\eta>0$ such that Inequality \eqref{eq_proof_raybis} holds for all $\eps\leq\eta$.
\end{proof}

\section{Two particular situations}\label{sect4}

Our goal in this section is to convince the reader that non-primitive QMS are more interesting (and less trivial) objects than their classical counterpart. We shall indeed highlight two manifestations of quantum phenomenons in our framework, the first one being the existence of quantum correlations in bipartite systems (Subsection \ref{sect41}) and the second one being quantum coherence (Subsection \ref{sect42}).

\subsection{The case where the Decoherence-Free algebra is a factor}\label{sect41}

In this subsection, we focus on the case where $\Hcal=\Hcal_A\otimes\Hcal_B$ and where $\DF=\Bcal(\Hcal_A)\otimes I_{\Hcal_B}$. We denote by $d_A$ and $d_B$ the respective dimensions of $\Hcal_A$ and $\Hcal_B$ and we assume that $d_B>1$. We write $H_A$ the Hamiltonian associated to the unitary evolution on system $A$ and $L$ the Lindbladian of the primitive Markovian evolution on system $B$, with unique faithful invariant state $\tau\in\Scal(\Hcal_B)$. Consequently, the evolution of the whole system is given by a QMS $\Pcal$ with generator
\begin{equation}\label{eq_lind_factor}
\Lcal(\cdot)=i[H_A\otimes I_{\Hcal_A},\cdot]+I\otimes L(\cdot)\,.
\end{equation}
The reference state is simply given by $\sigma_{\tr}=\frac{I_{\Hcal_A}}{d_A}\,\otimes\,\tau$, so that the conditional expectations are given by:
\[E_\Ncal(X)=\Tr_\tau[X]\otimes I_{\Hcal_B},\qquad E_{\Ncal*}(\rho)=\rho_A\otimes\tau\,,\]
where $\rho_A$ denotes the partial trace of the state $\rho$ with respect to $\Hcal_B$. Similarly, we write $\rho_B=\Tr_{\Hcal_A}[\rho]$. We recall that the modified log-Sobolev constant $\alpha_1(L)$ is defined as the best constant such that Inequality \eqref{eqdeflogsob1} holds. As $\Ncal(P)=\C I_{\Hcal_B}$, it is also true that $\alpha_1(L)=\alpha_\Ncal(L)$ (actually, everything that follows still holds if $P$ is not primitive and $\alpha_1(L)$ is replaced by $\alpha_\Ncal(L)$).

\begin{theo}\label{theo_complogsobolev}
With the above notations:
\begin{enumerate}
\item$\alpha_\Ncal(\Lcal)\leq\alpha_1(L)$.
\item Moreover, $P$ satisfies the \eqref{DFEP} if and only if $\Pcal$ does.
\item Consequently, $\alpha_\Ncal(\Lcal)>0$ if and only if $\alpha_1(L)>0$.
\end{enumerate}
\end{theo}

\begin{rem}
Point 3. in the above theorem is far from being obvious. For instance, it is not known if the $\Lbb_1$-regularity of $L$ implies the one of $\Lcal$. However, we know from Theorem \ref{theo_DFlogSob} that it implies that $\alpha_1(\Lcal)>0$. Thus, it could be true that $\Lcal$ fails to be $\Lbb_1$-regular whereas $\alpha_\Ncal(\Lcal)>0$. 
\end{rem}

\begin{proof}
We first prove that $\alpha_\Ncal(\Lcal)\leq\alpha_1(L)$. Let $\rho$ be a state on $\Hcal_B$. We have to prove that:
\[\alpha_\Ncal(\Lcal)\Dent{\rho}{\tau}\leq-\Tr\big[L_*(\rho)\left(\log\rho-\log\tau\right)\big]\,.\]
Let $\eta$ be any state on $\Hcal_A$. Then, as
\[\DentDF{\eta\otimes\rho}=\Dent{\eta\otimes\rho}{\eta\otimes\tau}=\Dent{\rho}{\tau}\,,\]
we have by definition of $\alpha_\Ncal(\Lcal)$
\[\alpha_\Ncal(\Lcal)\Dent{\rho}{\tau}\leq-\Tr\big[\Lcal_*(\eta\otimes\rho)\left(\log(\eta\otimes\rho)-\log(\frac{I_{\Hcal_A}}{d_A}\otimes\tau)\right)\big]\,.\]
To conclude, we show that the right-hand side of the two previous inequalities coincide.\\
Notice first that we have
\[\Lcal_*(\eta\otimes\rho )=-i[H_A,\eta]\otimes\rho+\eta\otimes L_*(\rho)\,.\]
Also notice that for $X\in\Bcal(\Hcal_A)$, $\Tr\big[[H_A,\eta]\,X\big]=0$ whenever $[\eta,X]=0$. Applying this to $X=\log\eta+(\log d_A) I_{\Hcal_A}$, we get
\begin{align*}
\Tr\big[\Lcal_*(\eta\otimes\rho)\left(\log(\eta\otimes\rho)-\log(\frac{I_{\Hcal_A}}{N}\otimes\tau)\right)\big]
& =\Tr\big[L_*(\rho)\left(\log\rho-\log\tau\right)\big]+\Tr[L_*(\rho)]\left(\Tr[\log\eta]+\log d_A\right)\\
&=\Tr\big[L_*(\rho)\left(\log\rho-\log\tau\right)\big]\,,
\end{align*}
where in the last step we used that $P_*$ is trace-preserving. This concludes the first part of the proof.\\

We now show the second part. By the first part, it is clear that if $\EP_L(\rho_B)=\EP_\Lcal(\rho_A\otimes\rho_B)$ for all $\rho_A\in\Scal(\Hcal_A)$, $\rho_B\in\Scal(\Hcal_B)$. This implies that $P$ satisfies the \eqref{DFEP} property whenever $\Pcal$ does. Suppose now that $P$ satisfies the \eqref{DFEP} and that $\EP_\Lcal(\rho)=0$ for a certain state $\rho\in\Scal(\Hcal)$. We need to show that $\rho=\rho_\Ncal=\rho_A\otimes\tau$. Consider the spectral decomposition of an observable $X\in\Bcal_{\text{sa}}(\Hcal_A)$: $X=\sum_k{x_k\,P_k}$ where $x_k\in\R$ and the $P_k$ form a family of mutually orthogonal projections that sum to the identity operator on $\Hcal_A$. We denote by $E_X$ the \emph{pinching} map associated to $X\otimes I_{\Hcal_B}$:
\[E_X\,:\,Y\in\Bcal(\Hcal)\mapsto E_X(Y)=\sum_k\,(P_k\otimes I_{\Hcal_B})\, Y\, (P_k\otimes I_{\Hcal_B})\,.\]
We claim that $\EP_\Lcal(\rho)\geq\EP_\Lcal(E_X(\rho))$. Indeed, leaving aside similar computations to what has already been done several times, we have
\[\DentDF{\rho}=\Dent{\rho}{E_X(\rho)}+\Dent{E_X(\rho)}{\rho_A\otimes\tau}\,.\]
Consequently:
\[\EP_\Lcal(\rho)-\EP_\Lcal(E_X(\rho))=\underset{t\to 0}{\lim}\,\frac{1}{t}\left(\Dent{\rho}{E_X(\rho)}-\Dent{\Pcal_{*t}(\rho)}{\Pcal_{*t}(E_X(\rho))}\right)\,.\]
The claim follows from the data-processing Inequality applied to the CP map $\Pcal_{*t}$. Consequently, we obtain that $\EP_\Lcal(E_X(\rho))=0$. It is a straightforward computation to show that this implies that $E_X(\rho)=\rho_A\otimes\tau$. Now, for all $X\in\Bcal(\Hcal_A)$ and $Y\in\Bcal(\Hcal_B)$, we get
\[\Tr[(X\otimes Y)\,\rho]=\Tr[(X\otimes Y)\,E_X(\rho)]=\Tr[(X\otimes Y)\,(\rho_A\otimes\tau)]\,.\]
As $\Bcal(\Hcal_A)\otimes\Bcal(\Hcal_B)=\Bcal(\Hcal_A\otimes\Hcal_B)$, this shows that $\rho=\rho_A\otimes\tau$, which proves our claim.\\

We now show the last part of the theorem. It is clear from the first part that $\alpha_\Ncal(\Lcal)>0$ implies that $\alpha_1(L)>0$, so assume that $\alpha_1(L)>0$. Consequently, $P$ satisfies the \eqref{DFEP} and by the second part, $\Pcal$ also does. Thus, in order to show that $\alpha_\Ncal(\Lcal)>0$, we can proceed as in the proof of Theorem \ref{theo_DFlogSob}. Take $\eps>0$, $\rho_A\in\Scal(\Hcal_A)$ and $Y\in\Bcal(\Hcal)$ with $E_\Ncal(Y)=0$, and define $\rho_\eps=\rho_A\otimes\tau+\eps\,\sigma_{\tr}^{\frac12}\,Y\,\sigma_{\tr}^{\frac12}$. As in the proof of the aforementioned theorem, we need to show that if $\Ecal_\Lcal\left(Y,\Theta_{\sigma}(\sigma_{\tr}^{\frac12}\,Y\,\sigma_{\tr}^{\frac12})\right)=0$, then necessarily $Y=0$. So assume that the first statement holds. Take $X\in\Bcal(\Hcal_A)$. As proved above, $\EP_\Lcal(\rho_\eps)\geq\EP_\Lcal(E_X(\rho_\eps))$, so that expanding both terms up to second order we obtain:
\begin{equation}\label{eq_proof_bipartite}
\Ecal_\Lcal\left(Y,\Theta_{\sigma}(\sigma_{\tr}^{\frac12}\,Y\,\sigma_{\tr}^{\frac12})\right)\geq\Ecal_\Lcal\left(E_X(Y),\Theta_{\sigma}(\sigma_{\tr}^{\frac12}\,E_X(Y)\,\sigma_{\tr}^{\frac12})\right)\,.
\end{equation}
Thus the right-hand side is null. However, it can be checked by a direct computation that for all $\rho\in\Scal(\Hcal)$,
\[\alpha_1(L)\,\DentDF{E_X(\rho)}\leq \EP_\Lcal\left(E_X(\rho)\right)\,,\]
or otherwise said, the DF log-Sobolev Inequality is satisfied with constant $\alpha_1(L)$ for states of the form $E_X(X)$ (that is, for classical-quantum states). As $\alpha_1(L)>0$, then necessarily the right-hand side of Inequality \eqref{eq_proof_bipartite} vanishes only on $\DF$. This again is a consequence of the proof of Theorem \ref{theo_DFlogSob}. Consequently $E_X(Y)\in\DF$. Remark now that $E_X$ and $E_\Ncal$ commute and therefore, for all $X\in\Bcal(\Hcal_A)$ and $Z\in\Bcal(\Hcal_B)$,
\[\Tr[(X\otimes Z)\,E_\Ncal(Y)]=\Tr[(X\otimes Z)\,E_X(E_\Ncal(Y))]=\Tr[(X\otimes Z)\,E_X(Y)]=\Tr[(X\otimes Z)\,Y]\,,\]
which shows that $Y\in\DF$. As we also have $E_\Ncal(Y)=0$, this concludes the proof.
\end{proof}

A question that remains is how much $\alpha_\Ncal(\Lcal)$ differ from $\alpha_1(L)$. Or, in other word, what is the best constant $C$ such that $C\,\alpha_1(L)\leq\alpha_\Ncal(\Lcal)$. Remark that in the classical case, the question is trivial as one can directly check that both constants are equal (they are also equal when the minimization is done over classical-quantum states, as mentioned in the above proof). We could not answer this question. However, we shall see that $\alpha_\Ncal(\Lcal)$ is intimately related to the rate of decay of the correlations between systems A and B.\\
\\As we shall see, the DF relative entropy and the entropy production are the sum of their trace on $B$ plus a term related to the mutual information of the state. Recall that the mutual information between systems $A$ and $B$ in the state $\rho\in\Scal(\Hcal)$ is given by:
\[\IM{\rho}=\Scal(\rho_A)+\Scal(\rho_B)-\Scal(\rho)=\Dent{\rho}{\rho_A\otimes\rho_B}\,.\]
This quantity quantifies the degree of correlation between the two subsystems $A$ and $B$, so in a sense it is natural that it should play a role in our context. We also define the \emph{Information Production (IP)} of the state $\rho\in\Scal(\Hcal)$ induced by $\Pcal$ as the quantity

\begin{equation}\label{eq_def_IP}
\IP_{\Lcal}(\rho)=-\frac{\partial}{\partial t}\, \IM{\Pcal_{*t}(\rho)}\,\Big|_{t=0}\,.
\end{equation}

\begin{prop}\label{prop_EPandIP}\ 
\begin{enumerate}
\item For any state $\rho\in\Scal(\Hcal)$, the DF entropy takes the form
\begin{equation}\label{eq_prop_EP}
\DentDF{\rho}=I_{\rho}(A:B)+\Dent{\rho_B}{\tau}\,.
\end{equation}
\item For any state $\rho\in\Scal(\Hcal)$, $\EP_\Lcal(\rho)=\IP_{\Lcal}(\rho)+\EP_L(\rho_B)$ and 
\begin{equation}\label{eq_prop_IP}
\IP_{\Lcal}(\rho)=-\Tr\left[\Lcal_*(\rho)\left(\log\,\rho-\log(\rho_A\otimes\rho_B)\right)\right]\,.
\end{equation}
\item $\IP_{\Lcal}\geq0$ and, consequently, $\EP_\Lcal(\rho)\geq\EP_L(\rho_B)$.
\end{enumerate}
\end{prop}

\begin{proof}
We prove 1.
\begin{align*}
\DentDF{\rho}
& =\Tr\left[\rho\left(\log\rho-\log(\rho_A\otimes\rho_B)\right)\right]+\Tr\left[\rho\left(\log(\rho_A\otimes\rho_B)-\log E_{\Ncal*}(\rho)\right)\right] \\
& = \IM{\rho}+\Tr\left[\rho\,\left(I_{\Hcal_A}\otimes(\log\rho_B-\log\tau)\right)\right] \\
& = \IM{\rho}+\Tr\left[\rho_B\left(\log\rho_B-\log\tau\right)\right]=\IM{\rho}+\Dent{\rho_B}{\tau}\,.
\end{align*}
The first part of point 2. is a direct consequence of point 1. We prove Equation \eqref{eq_prop_IP}: for all $\rho\in\Scal(\Hcal)$, we have 
\begin{align*}
& \IP_{\Lcal}(\rho) = \EP_\Lcal(\rho)-\EP_L(\rho_B)\,, \\
& \Tr\left[\Lcal_*(\rho)\log\sigma_{\tr}\right]=\Tr\left[L_*(\rho_B)\log\tau\right] \,, \\
& \Tr\left[\Lcal_*(\rho)\log(\rho_A\otimes\rho_B)\right]=\Tr\left[L_*(\rho_B)\log\rho_B\right]\,.
\end{align*}
The first equality is clear. The second and third equalities follows from the fact that $\Tr_{\Hcal_B}[\Lcal_*(\rho)]=0$. Indeed, for all $X\in\Bcal(\Hcal_A)$, 
\begin{align*}
\Tr\big[X\,\Tr_{\Hcal_B}[\Lcal_*(\rho)]\big] = \Tr\big[(X\otimes I_{\Hcal_B})\Lcal_*(\rho)\big] = \Tr\big[\Lcal_*\left((X\otimes I_{\Hcal_B})\,\rho\right)\big] =0\,,
\end{align*}
where in the before-last equality we used that $\Lcal_*=I\otimes L_*$ and where in the last equality we used that $\Pcal_*$ is trace-preserving. 
Putting these three equalities together shows the result.\\
To prove point 3. remark that
\[\IP_\Lcal(\rho)=\underset{t\to 0}{\lim}\,\frac{1}{t}\left(\IM{\rho}-\IM{\Pcal_{*t}(\rho)}\right)\,,\]
so that the result follows by the data-processing inequality applied to the mutual information and the CP map $\Pcal_{*t}$.
\end{proof}

It is very tempting at this point to define a functional inequality that implies the exponential decay of the mutual information. Thus, define $\beta_\Ncal(\Lcal)$ as the best constant fulfilling the following inequality: for all $\rho\in\Scal(\Hcal)$, it holds that
\begin{equation}\label{eq_def_IPlogsob}
2\,\beta\,\IM{\rho}\leq\IP_\Lcal(\rho)\,.
\end{equation}
It is now clear that this inequality holds for some $\beta\geq0$ if and only if for all $\rho\in\Scal(\Hcal)$,
\[\IM{\Pcal_{*t}(\rho)}\leq e^{-2\,\beta t}\,I_\rho(A;B)\,.\]

\begin{prop}\label{prop_bad_comp}
We have $\alpha_\Ncal(\Lcal)\geq\min\{\beta_\Ncal(\Lcal)\,,\,\alpha_1(L)\}$.
\end{prop}

\begin{proof}
This is clearly true as by Proposition \ref{prop_EPandIP},
\[2\,\min\{\beta_\Ncal(\Lcal)\,,\,\alpha_1(L)\}\,\DentDF{\rho}\leq 2\,\alpha_1(L)\,\Dent{\rho_A}{\tau}+2\,\beta_\Ncal(\Lcal)\,\IM{\rho}\leq\EP_\Lcal(\rho)\,.\]
\end{proof}

We leave the study of this inequality to further research.

\begin{ex}[The example of a partially depolarizing channel]
Let $\gamma$ be a positive constant and define, for a faithful state $\tau\in\Scal(\Hcal_B)$, the Lindbladian of the $\tau$-depolarizing channel in the Schrödinger picture acting on an observable $\rho_B\in\Scal(\Hcal_B)$ as:
\begin{equation}\label{eq_depol}
L_{\text{depol}*}^\gamma(\rho)=\gamma\left(\tau-\rho\right)\,.
\end{equation}
The whole evolution of the system $A+B$ is then given by the quantum Markov semigroup $(\Pcal_t^\gamma)_{t\geq0}$ whose Lindbladian in the Schrödinger picture acts on a state $\rho\in\Scal(\Hcal)$ as:
\begin{equation}\label{eq_depolAB}
\Lcal_{*}^\gamma(\rho)=\left(\Ind\otimes L_{\text{depol}*}^\gamma\right)(\rho)=\gamma\left(\rho_A\otimes\tau-\rho\right)\,.
\end{equation}
The depolarizing channel is a primitive QMS with unique invariant state $\tau$, so that $\Ncal(\Pcal^\gamma)=\Bcal(\Hcal_A)\otimes I_{\Hcal_B}$ and for any state $\rho\in\Scal(\Hcal)$, we have
\begin{equation}\label{eq_EIDdepol}
\Pcal_{*t}^\gamma(\rho)-\rho_A\otimes\tau \underset{t\to+\infty}{\longrightarrow}0\,.
\end{equation}
We can compute both the entropy production and the information production and we find, for $\rho\in\Scal(\Hcal)$:
\begin{align*}
& \EP_\Lcal(\rho)=\gamma\left(\Dent{\rho}{\rho_A\otimes\tau}+\Dent{\rho_A\otimes\tau}{\rho}\right)\gamma\,\Dent{\rho}{\rho_A\otimes\tau}\,,\\
& \IP_\Lcal(\rho)=\gamma\left(\IM{\rho}+\Dent{\rho_A\otimes\tau}{\rho}-\Dent{\tau}{\rho_B}\right)\geq \gamma\,\IM{\rho}\,.
\end{align*}
Consequently, we see that $\alpha_\Ncal(\Lcal^\gamma)\geq\gamma/2$ and similarly $\beta_\Lcal(\Lcal^\gamma)\geq\gamma/2$.
\end{ex}

\subsection{The case where the Decoherence-Free algebra is the algebra of diagonal operators}\label{sect42}

In this subsection, we focus on the case where $\DF$ is the algebra of diagonal operators: denoting by $(e_i)_{i=1,...,d}$ an orthonormal basis of $\Hcal$, we consider the case where 
\[\DF=\left\{X=\sum_{i=1}^d\,x_i\,\proj{e_i}{e_i}\ ;\ x_i\in\C\right\}\,.\]
Remark that necessarily the invariant states are diagonal in the o.n.b. $(e_i)$, so that $\sigma_{\tr}=\frac{I_\Hcal}{d}$, i.e. it is the maximally mixed state. As a consequence, $\Pcal$ is doubly stochastic and $\alpha_\Ncal(\Lcal)>0$. We can completely characterize this class of QMS.

\begin{prop}\label{prop_diag}
Assume that $\Pcal$ is a QMS on $\Bcal(\Hcal)$ with $\Hcal=\C^d$ and such that $\DF$ is isomorphic to the algebra of diagonal operators. We denote by $\Lcal$ the generator of $\Pcal$. With the above notation, there exists a selfadjoint $d\times d$-matrix $\Gamma=(\gamma_{i,j})_{i,j=1}^d$ such that $\gamma_{i,i}=0$, $\text{Re}\,\gamma_{i,j}<0$ for all $i\ne j$ and:
\begin{equation}\label{eq_prop_diag}
\Lcal(\proj{e_i}{e_j})=\gamma_{i,j}\,\proj{e_i}{e_j},\qquad\Pcal_t(\proj{e_i}{e_j})=e^{t\,\gamma_{i,j}}\proj{e_i}{e_j}\,.
\end{equation}
Moreover, $\Pcal$ is reversible if and only if $\Gamma$ can be chosen with real coefficients.
\end{prop}

\begin{proof}
We consider the form of the Lindbladian given in Equation \eqref{eqlindblad}. First notice that necessarily $\DF$ is equal to the set of fixed-points of $\Pcal$. Indeed, if there exists a one-parameter group of unitary operators $(U_t)_{t\in\R}$ on $\Hcal$ such that Equation \eqref{eq_unitary_evolution} holds, then necessarily $U^*_t\,\proj{e_i}{e_i}\,U_t=\proj{e_i}{e_i}$ for all time $t\geq0$ so that $\proj{e_i}{e_i}$ are all fixed-points for all $i$. Consequently, as proved for instance in \cite{DFSU1}, 
\[\DF=\{L_k,L_k^*\,;\,k=1,...,d\}'\,,\]
where the prime denotes the commutant of the set. It implies that the $L_k$ are all diagonal in the o.n.b. $(e_i)$: there exist $l_i(k)\in\C$ such that
\[L_k=\sum_k\,l_i(k)\,\proj{e_i}{e_i}\,.\]
Writing $\tilde\gamma_{i,j}=\sum_k\,\left(l_i(k)\,\overline{l_j(k)}-\frac{1}{2}\,(|l_i(k)|^2+|l_j(k)|^2)\right)$ and using Equation \eqref{eqlindblad} leads to
\[\Lcal(\proj{e_i}{e_j})=-\big[H,\,\proj{e_i}{e_j}\,\big]+\tilde\gamma_{i,j}\,\proj{e_i}{e_j}\,,\]
where $H$ is defined in Equation \eqref{eqlindblad}. For all $i$,
\[0=\Lcal(\proj{e_i}{e_i})=i\big[H,\,\proj{e_i}{e_i}\,\big]\,.\]
Consequently $H$ is diagonal in the o.n.b. $(e_i)$. We write $h_i$ the eigenvalue corresponding to the eigenvector $e_i$. Now, taking $\gamma_{ij}=\tilde \gamma_{ij}+i(h_i-h_j)$ leads to Equation \eqref{eq_prop_diag}. Remark furthermore that when $i\ne j$, then $\text{Re}\,\gamma_{ij}<0$ as necessarily $\lim_{t\to 0}\,\Pcal_t(\proj{e_i}{e_j})=0$ by Equation \eqref{eq_theo_deco2}.\\
\\If $\Pcal$ is reversible, we have
\[\frac{\overline{\gamma_{i,j}}}d=\sca{\Lcal(\proj{e_i}{e_j})}{\proj{e_j}{e_i}}=\sca{\proj{e_i}{e_j}}{\Lcal(\proj{e_j}{e_i})}=\frac{\gamma_{i,j}}d\,,\]
which proves that $\Gamma$ has real coefficients.
\end{proof}

Remark that the decoherence QMS corresponds to the case where all the $\gamma_{i,j}$ are equals. In general, by Theorem \ref{theocomparconst}, Theorem \ref{theo_DFlogSob} and Remark \ref{rem_theo_comp}, we obtain:

\begin{theo}\label{theo_diag}
If $\DF$ is isomorphic to the algebra of diagonal operators, then, with the notations of Proposition \ref{prop_diag},
\begin{equation}\label{eq_theo_diag1}
0<\alpha_\Ncal(\Lcal)\leq\lambda(\Lcal)=\min_{i,j}\,\text{Re}\,-\gamma_{i,j}\,.
\end{equation}
\end{theo}

In the case of the decoherence QMS, we find a better upper bound on the DF log-Sobolev constant in terms of the MLS constant of the depolarizing channel of the maximally mixed stated, computed in \cite{MSW16}. 

\begin{prop}\label{prop_DFLSIdeco}
One has
\begin{equation}\label{eq_prop_DFLSIdeco}
\alpha_\Ncal(\Lcal^\deco) \leq \alpha_1(L_{\text{depol}}^\gamma)\,.
\end{equation}
\end{prop}

\begin{proof}
Recall the definition of the depolarizing quantum Markov semigroup whose Lindbladian is given by
\[L_{\text{depol}}^\gamma(X)=\gamma\left(\Tr[X]\frac{I_{\Hcal}}{d}-X\right)\,.\]
By definition,
\[\alpha_1(L^\gamma_{\text{depol}})=\underset{\rho\in\Scal(\Hcal)}{\min}\,\frac{1}{2}\gamma\left(1+\frac{\Dent{\frac{I_\Hcal}{d}}{\rho}}{\Dent{\rho}{\frac{I_\Hcal}{d}}}\right)\,.\]
As $\Scal(\Hcal)$ is compact, the minimum is attained at a state that we denote by $\rho_{\min}$. We remark that, by invariance of the relative entropy under unitary conjugation of both states, only the spectrum on $\rho_{\min}$ matters and any state with the same spectrum achieved this minimum. We choose $\rho_{\min}$ so that it commutes with the Weyl unitary matrix:
\[U=\sum_{i=1}^d\,\proj{e_{i+1}}{e_i}\,,\]
where we set $e_{d+1}=e_1$. Otherwise stated, $\rho_{\min}$ is diagonal in the same o.n.b than $U$ and as a consequence, $E_{\Ncal*}(\rho_{\min})=\frac{I_\Hcal}{d}$. This implies
\[\alpha_\Ncal(\Lcal^\deco)\leq\frac{\EP_{\Lcal^\deco}(\rho_{\min})}{2\,\DentDF{\rho_{\min}}}=\frac{1}{2}\gamma\left(1+\frac{\Dent{\frac{I_\Hcal}{d}}{\rho}}{\Dent{\rho}{\frac{I_\Hcal}{d}}}\right)=\frac{\EP_{\Lcal_{\text{depol}}}(\rho_{\min})}{2\,\Dent{\rho_{\min}}{\frac{I_\Hcal}{d}}}=\alpha_1(L_{\text{depol}}^\gamma)\,.\]
\end{proof}

\section{Application to decoherence time for Quantum Markov Semigroups}\label{sect5}

We now state our main result concerning the decoherence time. Remark that, as in the primitive case, we do not ask the QMS to be reversible.

\begin{theo}\label{theo_decotime}
Let $\Pcal$ be a QMS on a finite dimensional Hilbert space $\Kcal$, with Lindbladian $\Lcal$. We assume that $\Pcal$ has a faithful invariant state and we denote by $\sigma_{\min}$ the smallest eigenvalue of $\sigma_{\tr}$.
\begin{itemize}
\item One obtain the following upper bound in terms of the DF Poincar\'e constant
\begin{equation}\label{eq_theo_decotime1}
\norm{\Pcal_{*t}\left(\rho-E_{\Ncal*}(\rho)\right)}_{\operatorname{Tr}}\leq\, \sqrt{1/\sigma_{\min}}\,e^{-\lambda(\Lcal)\,t}\,.
\end{equation}
\item The upper bound obtained using the DF log-Sobolev constant is
\begin{equation}\label{eq_theo_decotime2}
\norm{\Pcal_{*t}\left(\rho-E_{\Ncal*}(\rho)\right)}_{\operatorname{Tr}}\leq\, \sqrt{2\log\left(1/\sigma_{\min}\right)}\, e^{-\alpha_\Ncal(\Lcal)\,t}\,.
\end{equation}
\end{itemize}
\end{theo}

\begin{proof}
The second Inequality \eqref{eq_theo_deco2} is directly a consequence of Theorem \ref{theoDFentropy}, Pinsker's Inequality and the fact that
\[\DentDF{\rho}\,\leq\,\Dent{\rho}{\sigma_{\tr}}\leq -\log(\sigma_{\min}),\qquad\forall\rho\in\Scal(\Hcal)\,,\]
where the first inequality comes from Lemma \ref{lem_DFentropy} and the second can be proved directly using that $\sigma_{\tr}\geq\sigma_{\min}\, I_\Hcal$ and the operator monotonicity of the logarithm.\\
In order to prove the first Inequality \eqref{eq_theo_decotime1}, we need to introduce the $\chi_2$-divergence between two states defined as
\[\chi^2(\rho,\sigma)=\Tr\left[(\rho-\sigma)\,\sigma^{-\frac{1}{2}}\,(\rho-\sigma)\,\sigma^{-\frac{1}{2}}\right],\qquad\rho,\sigma\in\Scal(\Hcal)\,.\]
Then, it is proved in \cite{Rus94} Theorem 3 that $\norm{\rho-\sigma_{\tr}}_1^2\leq\chi_2(\rho,\sigma_{\tr})\leq 1/\sigma_{\min}$ (see also \cite{TKRW} for a generalization of this result to more general forms of divergences). Besides, a simple computation shows that
\[\chi^2\left(\rho,E_{\Ncal*}(\rho)\right)=\Var_\Ncal(\sigma_{\tr}^{-\frac{1}{2}}\,\rho\,\sigma_{\tr}^{-\frac{1}{2}})\,.\]
As, by Lemma \ref{lemDFvariance}, $\Var_\Ncal(X)\leq\Var_{\sigma_{\tr}}(X)$ for any $X\in\Bcal(\Hcal)$, Inequality \eqref{eq_theo_decotime1} follows by Theorem \ref{theo_speedEIDvar}
\end{proof}

\begin{example}
We conclude this section with a generalization of the examples discussed so far. Consider any $*$-algebra $\Ncal$ of $\Bcal(\Hcal)$ with $\Hcal=\C^d$ and denote by $E_\Ncal$ any faithful conditional expectation from $\Bcal(\Hcal)$ to $\Ncal$. We define $\sigma_{\tr}=\frac{1}{d}\Tr\circ E_\Ncal$. The conditional expectation $E_\Ncal$ is consequently the orthogonal projection on $\Ncal$ for the $\sca{\cdot}{\cdot}_{\sigma_{\Tr}}$ scalar product. Similarly to \eqref{eq_def_QMSdeco} and \eqref{eq_depolAB}, we can define on $\Hcal$ the QMS with Lindbladian:
\[\Lcal_{\Ncal}^\gamma(X)=\gamma\left(E_{\Ncal}(X)-X\right)\,.\]
This defines a reversible Lindbladian with respect to the state $\sigma_{\tr}$.
Then, by usual computations, $\lambda(\Lcal_\Ncal)=\gamma$ and
\[\EP_{\Lcal_\Ncal^\gamma}(\rho)=\gamma\left(\Dent{\rho}{\rho_\Ncal}+\Dent{\rho_\Ncal}{\rho}\right)\,,\]
and consequently $0\leq\frac{\gamma}{2}\leq\alpha_\Ncal(\Lcal_\Ncal)\leq\gamma$.\newline
\\Assume now that $\sigma_{\tr}$ is the maximally-mixed state: $\sigma_{\tr}=\frac{I_\Hcal}{d}$. Going back to the definition of the decoherence-time \eqref{eq_def_decotime2}, we obtain with the spectral gap the estimate:
\[\tau_\deco\geq\frac{1}{2\gamma}\,\log \left(d\,\eps^{-2}\right)\,,\]
and with the DF log-Sobolev constant:
\[\tau_\deco\geq\,\frac{1}{\gamma}\log \left(2\log d\,\eps^{-2}\right)\,.\]
We see that with the first estimate, $\tau_\deco=\Omega (\log d)$ while with the second estimate, $\tau_\deco=\Omega (\log \log d)$, which improves the estimate by a logarithm. As emphasized in \cite{KT2013}, such an improvement is possible only when the modified DF log-Sobolev constant is of the same order as the spectral gap.
\end{example}

\section{Conclusion}\label{sect6}

In this article we proposed a natural framework in order to study functional inequalities for not necessarily primitive quantum Markov semigroups. We illustrated this framework by studying generalizations of the Poincaré Inequality and the modified log-Sobolev Inequality. Such inequalities imply the rapid decoherence of the quantum Markov semigroup.

\paragraph{}Some new difficulties appear in this context relative to the quantum nature of the system and it seems more challenging to obtain estimates of the new log-Sobolev constant. We highlighted this fact in the case of a bipartite system where only one of the system undergoes a irreversible evolution. In this scenario, we introduced the production of information of the QMS, that we showed was a relevant quantity in order to estimate the DF log-Sobolev constant.

\paragraph{}We believe that the practical interest of this generalization can motivate the study of such inequalities. One natural continuation would be to generalize the notion of Hypercontractivity and its equivalence with Gross' log Sobolev Inequality (a similar direction has recently been proposed in \cite{BK16}). This will be the subject of a future article \cite{BR17}.

\paragraph{Acknowledgement:} The author is thankful to many persons for their advises and fruitful discussions, among which Raffaella Carbone in Pavia, and Nilanjana Datta and Cambyse Rouz\'e in Cambridge.

\bibliographystyle{abbrv}
\bibliography{biblio}

\begin{thebibliography}{10}

\bibitem{AFR14}
J.~Agredo, F.~Fagnola, and R.~Rebolledo.
\newblock {Decoherence free subspaces of a quantum Markov semigroup}.
\newblock {\em Journal of Mathematical Physics}, 55(11):112201, 2014.

\bibitem{Bak94}
D.~Bakry.
\newblock {L'hypercontractivit{\'e} et son utilisation en th{\'e}orie des
  semigroupes}.
\newblock {\em Lectures on probability theory}, pages 1--114, 1994.

\bibitem{BR17}
I.~Bardet and C.~Rouz{\'e}.
\newblock {The logarithmic Sobolev Inequality for non-primitive quantum Markov
  semigroups and estimation of decoherence rates}.
\newblock {\em in preparation}.

\bibitem{BK16}
S.~Beigi and C.~King.
\newblock {Hypercontractivity and the logarithmic Sobolev inequality for the
  completely bounded norm}.
\newblock {\em Journal of Mathematical Physics}, 57(1):015206, 2016.

\bibitem{B-O}
P.~Blanchard and R.~Olkiewicz.
\newblock {Decoherence induced transition from quantum to classical dynamics}.
\newblock {\em Reviews in Mathematical Physics}, 15(03):217--243, 2003.

\bibitem{BT06}
S.~G. Bobkov and P.~Tetali.
\newblock {Modified logarithmic Sobolev inequalities in discrete settings}.
\newblock {\em Journal of Theoretical Probability}, 19(2):289--336, 2006.

\bibitem{C-M}
R.~Carbone and A.~Martinelli.
\newblock {Logarithmic Sobolev inequalities in non-commutative algebras}.
\newblock {\em Infinite Dimensional Analysis, Quantum Probability and Related
  Topics}, 18(02):1550011, 2015.

\bibitem{C-S}
R.~Carbone and E.~Sasso.
\newblock {Hypercontractivity for a quantum {O}rnstein-{U}hlenbeck semigroup}.
\newblock {\em Probab. Theory Related Fields}, 140(3-4):505--522, 2008.

\bibitem{CSU2}
R.~Carbone, E.~Sasso, and V.~Umanit{\`a}.
\newblock {Decoherence for positive semigroups on {$M_2(\Bbb C)$}}.
\newblock {\em J. Math. Phys.}, 52(3):032202, 17, 2011.

\bibitem{CSU3}
R.~Carbone, E.~Sasso, and V.~Umanit{\`a}.
\newblock {Decoherence for quantum {M}arkov semi-groups on matrix algebras}.
\newblock {\em Ann. Henri Poincar{\'e}}, 14(4):681--697, 2013.

\bibitem{CSU4}
R.~Carbone, E.~Sasso, and V.~Umanit{\`a}.
\newblock {Environment induced decoherence for Markovian evolutions}.
\newblock {\em Journal of Mathematical Physics}, 56(9), 2015.

\bibitem{CM16}
E.~A. Carlen and J.~Maas.
\newblock {Gradient flow and entropy inequalities for quantum Markov semigroups
  with detailed balance}.
\newblock {\em Journal of Functional Analysis}, 2017.

\bibitem{CS03}
F.~Cipriani and J.-L. Sauvageot.
\newblock {Derivations as square roots of Dirichlet forms}.
\newblock {\em Journal of Functional Analysis}, 201(1):78--120, 2003.

\bibitem{CT15}
G.~I. Cirillo and F.~Ticozzi.
\newblock {Decompositions of Hilbert spaces, stability analysis and convergence
  probabilities for discrete-time quantum dynamical semigroups}.
\newblock {\em Journal of Physics A: Mathematical and Theoretical},
  48(8):085302, 2015.

\bibitem{CKMT15}
T.~Cubitt, M.~Kastoryano, A.~Montanaro, and K.~Temme.
\newblock {Quantum reverse hypercontractivity}.
\newblock {\em Journal of Mathematical Physics}, 56(10):102204, 2015.

\bibitem{CLMP-G13}
T.~S. Cubitt, A.~Lucia, S.~Michalakis, and D.~Perez-Garcia.
\newblock {Stability of local quantum dissipative systems}.
\newblock {\em Communications in Mathematical Physics}, 337(3):1275--1315,
  2015.

\bibitem{DFSU1}
J.~Deschamps, F.~Fagnola, E.~Sasso, and V.~Umanita.
\newblock {Structure of uniformly continuous quantum Markov semigroups}.
\newblock {\em Reviews in Mathematical Physics}, 28(01):1650003, 2016.

\bibitem{DSC96}
P.~Diaconis, L.~Saloff-Coste, et~al.
\newblock {Logarithmic Sobolev inequalities for finite Markov chains}.
\newblock {\em The Annals of Applied Probability}, 6(3):695--750, 1996.

\bibitem{FV07}
F.~Fagnola and V.~Umanit{\`a}.
\newblock {Generators of detailed balance quantum Markov semigroups}.
\newblock {\em Infinite Dimensional Analysis, Quantum Probability and Related
  Topics}, 10(03):335--363, 2007.

\bibitem{GIS}
D.~Guido, T.~Isola, and S.~Scarlatti.
\newblock {Non-symmetric Dirichlet Forms on Semifinite von Neumann Algebras}.
\newblock {\em Journal of Functional Analysis}, 135(1):50--75, 1996.

\bibitem{H11}
M.~Hellmich.
\newblock {Quantum dynamical semigroups and decoherence}.
\newblock {\em Advances in Mathematical Physics}, 2011, 2011.

\bibitem{JPW14}
V.~Jak\v{s}i{\'c}, C.-A. Pillet, and M.~Westrich.
\newblock {Entropic fluctuations of quantum dynamical semigroups}.
\newblock {\em Journal of Statistical Physics}, 154(1-2):153--187, 2014.

\bibitem{KRS11}
M.~J. Kastoryano, F.~Reiter, and A.~S. S{\o}rensen.
\newblock {Dissipative preparation of entanglement in optical cavities}.
\newblock {\em Physical review letters}, 106(9):090502, 2011.

\bibitem{KT2013}
M.~J. {Kastoryano} and K.~{Temme}.
\newblock {Quantum logarithmic Sobolev inequalities and rapid mixing}.
\newblock {\em Journal of Mathematical Physics}, 54(5):052202, May 2013.

\bibitem{KBLW01}
J.~Kempe, D.~Bacon, D.~A. Lidar, and K.~B. Whaley.
\newblock {Theory of decoherence-free fault-tolerant universal quantum
  computation}.
\newblock {\em Physical Review A}, 63(4):042307, 2001.

\bibitem{KFGV77}
A.~Kossakowski, A.~Frigerio, V.~Gorini, and M.~Verri.
\newblock {Quantum detailed balance and KMS condition}.
\newblock {\em Communications in Mathematical Physics}, 57(2):97--110, 1977.

\bibitem{KBDKMZ08}
B.~Kraus, H.~P. B{\"u}chler, S.~Diehl, A.~Kantian, A.~Micheli, and P.~Zoller.
\newblock {Preparation of entangled states by quantum Markov processes}.
\newblock {\em Physical Review A}, 78(4):042307, 2008.

\bibitem{LCW98}
D.~A. Lidar, I.~L. Chuang, and K.~B. Whaley.
\newblock {Decoherence-free subspaces for quantum computation}.
\newblock {\em Physical Review Letters}, 81(12):2594, 1998.

\bibitem{LW03}
D.~A. Lidar and K.~B. Whaley.
\newblock {Decoherence-free subspaces and subsystems}.
\newblock In {\em {Irreversible Quantum Dynamics}}, pages 83--120. Springer,
  2003.

\bibitem{Lind}
G.~Lindblad.
\newblock {On the generators of quantum dynamical semigroups}.
\newblock {\em Comm. Math. Phys.}, 48(2):119--130, 1976.

\bibitem{LCMP-G15}
A.~Lucia, T.~S. Cubitt, S.~Michalakis, and D.~P{\'e}rez-Garc{\'i}a.
\newblock {Rapid mixing and stability of quantum dissipative systems}.
\newblock {\em Physical Review A}, 91(4):040302, 2015.

\bibitem{Mart99}
F.~Martinelli.
\newblock {\em {Lectures on Glauber Dynamics for Discrete Spin Models}}, pages
  93--191.
\newblock Springer Berlin Heidelberg, Berlin, Heidelberg, 1999.

\bibitem{MO94}
F.~Martinelli and E.~Olivieri.
\newblock {Approach to equilibrium of Glauber dynamics in the one phase
  region}.
\newblock {\em Communications in Mathematical Physics}, 161(3):447--486, 1994.

\bibitem{MSW16}
A.~M{\"u}ller-Hermes, D.~S. Fran\c{c}a, and M.~M. Wolf.
\newblock {Relative entropy convergence for depolarizing channels}.
\newblock {\em Journal of Mathematical Physics}, 57(2):022202, 2016.

\bibitem{MSW2016}
A.~M{\"u}ller-Hermes, D.~{Stilck Fran\c{c}a}, and M.~M. Wolf.
\newblock {Entropy production of doubly stochastic quantum channels}.
\newblock {\em Journal of Mathematical Physics}, 57(2):022203, 2016.

\bibitem{O-Z}
R.~Olkiewicz and B.~Zegarlinski.
\newblock {Hypercontractivity in noncommutative {$L_p$} spaces}.
\newblock {\em J. Funct. Anal.}, 161(1):246--285, 1999.

\bibitem{Rus94}
M.~B. Ruskai.
\newblock {Beyond strong subadditivity? Improved bounds on the contraction of
  generalized relative entropy}.
\newblock {\em Reviews in Mathematical Physics}, 6(05a):1147--1161, 1994.

\bibitem{Spo1}
H.~Spohn.
\newblock {Entropy production for quantum dynamical semigroups}.
\newblock {\em Journal of Mathematical Physics}, 19(5):1227--1230, 1978.

\bibitem{SL78}
H.~Spohn and J.~L. Lebowitz.
\newblock {Irreversible thermodynamics for quantum systems weakly coupled to
  thermal reservoirs}.
\newblock {\em Adv. Chem. Phys}, 38:109--142, 1978.

\bibitem{TKRW}
K.~Temme, M.~J. Kastoryano, M.~B. Ruskai, M.~M. Wolf, and F.~Verstraete.
\newblock {The {$\chi^2$}-divergence and mixing times of quantum {M}arkov
  processes}.
\newblock {\em J. Math. Phys.}, 51(12):122201, 19, 2010.

\bibitem{TPK}
K.~{Temme}, F.~{Pastawski}, and M.~J. {Kastoryano}.
\newblock {Hypercontractivity of quasi-free quantum semigroups}.
\newblock {\em Journal of Physics A Mathematical General}, 47:5303, Oct. 2014.

\bibitem{VWC09}
F.~Verstraete, M.~M. Wolf, and J.~I. Cirac.
\newblock {Quantum computation and quantum-state engineering driven by
  dissipation}.
\newblock {\em Nature physics}, 5(9):633--636, 2009.

\bibitem{wolftour}
M.~M. Wolf.
\newblock {Quantum channels \& operations: Guided tour}.
\newblock
  \url{http://www-m5.ma.tum.de/foswiki/pub/M5/Allgemeines/MichaelWolf/QChannelLecture.pdf},
  2012.
\newblock Lecture notes based on a course given at the Niels-Bohr Institute.

\bibitem{Zeg90log2}
B.~Zegarlinski.
\newblock {Log-Sobolev inequalities for infinite one dimensional lattice
  systems}.
\newblock {\em Communications in mathematical physics}, 133(1):147--162, 1990.

\bibitem{Zeg90log}
B.~Zegarlinski.
\newblock {On log-Sobolev inequalities for infinite lattice systems}.
\newblock {\em Letters in Mathematical Physics}, 20(3):173--182, 1990.

\bibitem{Zeg92}
B.~Zegarlinski.
\newblock {Dobrushin uniqueness theorem and logarithmic Sobolev inequalities}.
\newblock {\em Journal of functional analysis}, 105(1):77--111, 1992.

\bibitem{Zur81}
W.~H. Zurek.
\newblock {Pointer basis of quantum apparatus: Into what mixture does the wave
  packet collapse?}
\newblock {\em Physical Review D}, 24(6):1516, 1981.

\bibitem{Zur82}
W.~H. Zurek.
\newblock {Environment-induced superselection rules}.
\newblock {\em Physical Review D}, 26(8):1862, 1982.

\end{thebibliography}

\end{document}